%% file: main.tex
\documentclass[journal,twoside,web]{IEEEtran}
\usepackage{generic}
\usepackage{amsmath,amsfonts,mathtools,amssymb, amsthm}
\usepackage{graphicx}
\usepackage{textcomp}
\usepackage{tikz}
\usetikzlibrary{arrows.meta,calc,decorations.markings,math}
\usepackage{pgfplots}
\usepackage{aligned-overset}
\usepackage{nicefrac}
\usepackage{algorithm}
\usepackage[shortlabels]{enumitem}
\usepackage{soul}
\usepackage{siunitx}
\DeclareSIUnit{\byte}{B}
\usetikzlibrary{positioning}
\usepackage[hidelinks, pdfusetitle]{hyperref}
\allowdisplaybreaks
\usepackage[noend]{algpseudocode}
\usepackage{yfonts}
\usepackage{tabularx}
\usepackage{cite}

\newtheorem{lemma}{\bf Lemma}
\newtheorem{theorem}{\bf Theorem}
\newtheorem{proposition}{\bf Proposition}
\newtheorem{corollary}{\bf Corollary}

\newtheorem{assumption}{\bf Assumption}

\newtheorem{remark}{\bf Remark}

\usepackage{comment}

\input{macros.tex}

\begin{document}
\title{
Automatic nonlinear MPC approximation with closed-loop guarantees}
\author{Abdullah Tokmak$^{1,2,3}$,
Christian Fiedler$^{2}$,
Melanie N. Zeilinger$^3$,
Sebastian Trimpe$^{2}$,
Johannes Köhler$^3$
\thanks{This work has been supported by the Swiss National Science Foundation
under NCCR Automation (grant agreement 51NF40 180545),
the IDEA League,
and the German Academic Scholarship Foundation.
}
\thanks{$^1$%
Department of Electrical Engineering and Automation, 
Aalto University,
02150 Espoo,
Finland,
\texttt{abdullah.tokmak@aalto.fi}.
}
\thanks{$^2$%
Institute for Data Science in Mechanical Engineering, RWTH Aachen University, 52068 Aachen,
Germany,
\texttt{\{christian.fiedler, trimpe\}@dsme.rwth-aachen.de}.
}
\thanks{$^3$%
Institute for Dynamic Systems and Control, ETH Zurich, Zurich CH-8092, Switzerland,
\texttt{\{mzeilinger, jkoehle\}@ethz.ch}.
}
}

\maketitle

\input{Sections/abstract.tex}
\input{Sections/introduction.tex}
\input{Sections/preliminaries.tex}
\input{Sections/alki.tex}

\input{Sections/alkiax.tex}

\input{Sections/MPC.tex}
\input{Sections/numerical_experiment.tex}
\input{Sections/summary_outlook.tex}

\input{Sections/acknowledgements.tex}
\input{Sections/appendix.tex}


\bibliography{references_TAC}
\bibliographystyle{IEEEtran}

\input{Sections/biographies.tex}
\end{document}

%% file: macros.tex
\newcommand{\swname}[1]{\texttt{#1}}

\newcommand*{\python}{\swname{Python}}

\newcommand{\ie}{i\/.\/e\/.,\/~}
\newcommand{\eg}{e\/.\/g\/.,\/~}

\newcommand{\domain}{\mathcal{X}}

\newcommand{\error}{{\epsilon}}

\newcommand{\cubes}{\mathcal{C}}

\newcommand{\alkiax}{\mbox{\textsc{Alkia-x}}}

\newcommand{\casadi}{CasADi}

\newcommand{\xcm}{\overline{x}_{c}}
\newcommand{\kla}{
{\widetilde{k}_{a}}
}

\newcommand{\knorm}{\mathfrak{K}}
\newcommand{\kinv}{\mathfrak{K}^{-1}}
\newcommand{\xfeasf}{\mathcal{X}_\mathrm{feas}^f}
\newcommand{\xfeash}{\mathcal{X}_\mathrm{feas}^h}
\newcommand{\rkhsx}{\overline{\Gamma}_a}
\newcommand{\rkhso}{\overline{\Gamma}_a}

\newcommand{\kernelcite}{
Scharnhorst2022Robust,
Maddalena2021Deterministic,
Steinwart2008SVM, Fasshauer2015Kernel, Kanagawa2018Gaussian,
Binder2019GPs, Rose2023Learning,
Rasmussen2006Gaussian, Liu2020Scalable, 
Lederer2020Realtime, Nguyen2008Local, Berkenkamp2021Bayesian, Srinivas2010Gaussian, Chowdhury2017Kernelized, Fiedler2021Practical, Sui2015Safe, Hashimoto2020Nonlinear}

\newcommand{\GPcite}{Kanagawa2018Gaussian, 
Binder2019GPs, Rose2023Learning, Rasmussen2006Gaussian, Liu2020Scalable, 
Lederer2020Realtime, Nguyen2008Local,
Berkenkamp2021Bayesian, Srinivas2010Gaussian, Chowdhury2017Kernelized, Fiedler2021Practical, Sui2015Safe, Hashimoto2020Nonlinear}

\newcommand{\RKHSnormcite}{Scharnhorst2022Robust, Maddalena2021Deterministic, Berkenkamp2021Bayesian, Srinivas2010Gaussian, Chowdhury2017Kernelized, Fiedler2021Practical, Sui2015Safe, Hashimoto2020Nonlinear}

\newcommand{\KRRcite}{Scharnhorst2022Robust, Maddalena2021Deterministic, Steinwart2008SVM, Fasshauer2015Kernel, Kanagawa2018Gaussian}

\newcommand{\kernelinterpolationcite}{Steinwart2008SVM, Fasshauer2015Kernel, Maddalena2021Deterministic, Kanagawa2018Gaussian}

\newcommand{\robustMPCcite}{Houska2019Robust, Kohler2021Uncertain, Sasfi2023Robust, Rakovic2023Robust}

\newcommand{\AMPCcite}{ 
Chen2018Approx,
Karg2020Efficient,
Chen2019Large,
Hose2023Approximate,
Zhang2021Policy,
Karg2021Probabilistic,
Hertneck2018Learning,
Nubert2020Robot, 
Pin2013Approximate,
parisini1998nonlinear,
gonzalez2023neural,
Canale2009Lipschitz,
Canale2010Lipschitz,
Canale2014Nonlinear,
boggio2022set,
Ganguly2023Interpolation,
Binder2019GPs,
Rose2023Learning}

%% file: Sections/abstract.tex
\begin{abstract}
 Safety guarantees are vital in many control applications, such as robotics.
Model predictive control (MPC) provides a constructive framework for controlling safety-critical systems, but is limited by its computational complexity.
We address this problem by presenting a novel algorithm that automatically computes an explicit approximation to nonlinear MPC schemes while retaining closed-loop guarantees.
Specifically, the problem can be reduced to a function approximation problem, which we then tackle by proposing \alkiax, the {A}daptive and {L}ocalized {K}ernel {I}nterpolation {A}lgorithm with e{X}trapolated reproducing kernel Hilbert space norm.
\alkiax\ is a non-iterative algorithm that ensures numerically well-conditioned computations, a fast-to-evaluate approximating function, and the guaranteed satisfaction of any desired bound on the approximation error.
Hence, \alkiax\ automatically computes an explicit function that approximates the MPC, yielding a controller suitable for safety-critical systems and high sampling rates.
We apply \alkiax\ to approximate two nonlinear MPC schemes, demonstrating reduced computational demand and applicability to realistic problems.
\end{abstract}


%% file: Sections/introduction.tex
\section{Introduction}
\label{sec:introduction}
\IEEEPARstart{M}{odel}
predictive control (MPC)~\cite{Rawlings2017Model} is an optimization-based control method for nonlinear constrained systems, in which the applied control input is implicitly defined by the solution of a nonlinear program (NLP). 
Hence, online application of MPC  under real-time requirements is challenging for fast applications like robotics.
This has motivated extensive research on approximate MPC schemes~\cite{\AMPCcite} to obtain controllers suitable for high sampling rates.
Specifically, these approaches approximate the MPC offline, \ie prior to applying the controller.
This offline approximation should ideally be automatic 
to simplify its application. 
Furthermore, the approximate MPC should retain closed-loop guarantees regarding stability and constraint satisfaction to allow for application on safety-critical systems.

Figure~\ref{fig:framework} outlines the proposed approach. 
We consider a \textit{robust} MPC formulation~\cite{\robustMPCcite}, which is designed such that the desired closed-loop guarantees remain valid under input disturbances below a user-chosen error bound $\error>0$.
By approximating the feedback law implicitly defined by this MPC up to a tolerance~$\error$, the approximate MPC preserves all control-theoretic guarantees induced by the
MPC.
Hence, approximating the MPC can be cast as a function approximation problem by sampling state and corresponding optimal inputs obtained by solving the MPC offline.
To address this function approximation problem,
we propose \alkiax, the \textbf Adaptive and \textbf Localized \textbf Kernel \textbf Interpolation \textbf Algorithm with e\textbf{X}trapolated reproducing kernel Hilbert space (RKHS) norm.
\alkiax\ automatically computes an explicit function that approximates the MPC with a uniform approximation error~$\error$, resulting in a cheap-to-evaluate approximate MPC with guarantees on stability and constraint satisfaction.

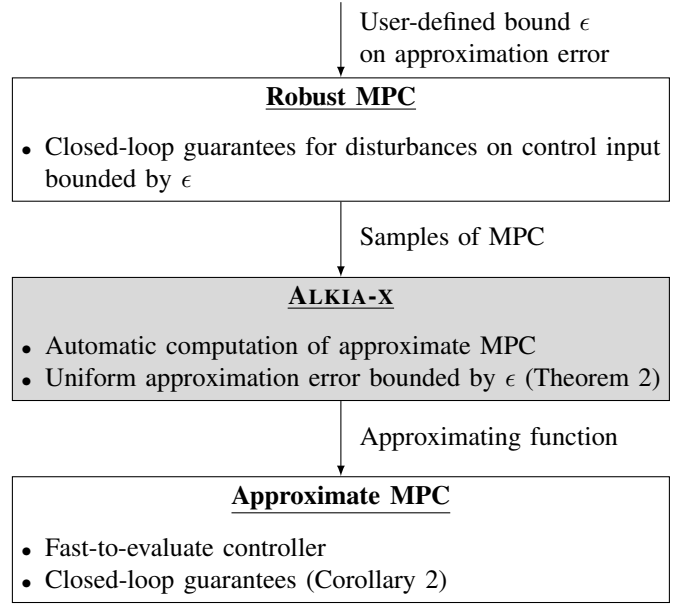
\begin{figure}[t]
    \centering
    \begin{center}
    \input{figures_final/figure_intro}
    \end{center}
    \caption{
    A high level illustration of the proposed framework for approximate MPC with closed-loop guarantees.  
    \alkiax\ automatically computes an explicit function that approximates a robust MPC scheme up to any specified accuracy~$\error$, which yields a fast-to-evaluate approximate MPC with closed-loop guarantees.
    Notably, the proposed process is \emph{non-iterative}
    due to the uniform error bounds inherently guaranteed by \alkiax.
    }
    \label{fig:framework}
\end{figure}

\subsubsection*{Related work}
\label{sec:related_work}
\color{black}
Linear MPC schemes can be exactly characterized as an explicit piecewise affine function over a polyhedral partition of the domain, which is known as explicit MPC~\cite{Bemporad2002Explicit, Alessio2009Explicit}. 
However, this exact solution tends to suffer from scalability issues, which has motivated research on approximate explicit MPC schemes.  
In~\cite{Johansen2003Tree}, the domain is partitioned into hypercubes and the MPC is approximated using piecewise affine functions. 
Notably, \cite[Lemma 1]{Johansen2003Tree} ensures that the approximate MPC respects the constraints by using a constrained optimization problem that leverages convexity. 
A similar approximate MPC is suggested in~\cite{Summers2011Multiresolution}, using a more general continuous interpolation on each hypercube.
However, these approaches primarily rely on convexity, and hence cannot be applied to nonlinear MPC schemes.

An overview of explicit MPC for nonlinear systems is given in~\cite{Grancharova2012Explicit}.
A very appealing approach is to treat the MPC as a black-box function and use supervised machine learning techniques to obtain an explicit approximation. 
Approximating MPC schemes with neural networks (NNs) is a particularly popular approach~\cite{Chen2018Approx,Karg2020Efficient,Chen2019Large,Hose2023Approximate,Zhang2021Policy,Karg2021Probabilistic,Hertneck2018Learning,Nubert2020Robot, Pin2013Approximate}, which was first suggested in~\cite{parisini1998nonlinear}.
However, it is difficult to ensure that the resulting NN provides desirable closed-loop properties, see the recent review paper~\cite{gonzalez2023neural}. 
One solution to this problem is adding a safety filter to change the output of the NN~\cite{Chen2018Approx,Karg2020Efficient,Chen2019Large,Hose2023Approximate,Zhang2021Policy}. 
For linear systems, this can be achieved with a projection on a robust control invariant set~\cite{Chen2018Approx,Karg2020Efficient} or by using additional active set iterations~\cite{Chen2019Large}.
For general nonlinear systems, an implicit safety filter can be obtained by validating the safety of the NN policy online and using an MPC-intrinsic fallback in case of predicted constraint violations~\cite{Hose2023Approximate}, see also~\cite{Zhang2021Policy}. 
Other solutions include a posteriori safety validations of the NN using statistical validation~\cite{Karg2021Probabilistic}.

A constructive approach to obtain closed-loop guarantees under approximation errors is to leverage \textit{robust} MPC design methods~\cite{\robustMPCcite}.
For instance, \cite{Pin2013Approximate,Hertneck2018Learning,Nubert2020Robot} first design an MPC scheme that is robust with respect to input disturbances bounded by~$\error$ and then use NNs to approximate this MPC with an approximation error smaller than~$\error$.
However, guaranteeing the desired approximation error bound for NNs is challenging. 
In~\cite{Hertneck2018Learning,Nubert2020Robot}, this issue is addressed by using statistical tools to validate the approximation accuracy with high probability. 
Nevertheless, this separate validation yields a nontrivial offline design consisting of multiple iterations between training the NN and validating its approximation.

A natural solution to this problem are nonparametric regression tools, such as kernel-based methods~\cite{\kernelcite}
or set-membership identification~%
\cite{Calliess2014PhD,milanese2004set,Canale2009Lipschitz,Canale2010Lipschitz, Canale2014Nonlinear,boggio2022set}, which both provide error bounds. 
Nonparametric set-membership estimation typically assumes Lipschitz continuity and 
builds a non-falsified piecewise affine set that contains the ground truth, which is also called Kinky inference~\cite{Calliess2014PhD} or Lipschitz interpolation~\cite{milanese2004set}. 
Applications of such tools to approximate nonlinear MPC schemes are explored in~\cite{Canale2009Lipschitz, Canale2010Lipschitz, Canale2014Nonlinear,boggio2022set}.
Similarly, in~\cite{Ganguly2023Interpolation}, a known Lipschitz bound is utilized to approximate an MPC with guaranteed error bound~$\error$, using instead a 
quasi-interpolation approach.
Kernel-based methods, such as 
kernel ridge regression (KRR)~\cite{\KRRcite}
or
Gaussian process (GP) regression~\cite{\GPcite},
enable smooth approximations and error bounds, assuming the ground truth lies in the corresponding RKHS.  
Applications of GPs to approximate MPC schemes are, \eg explored in~\cite{Binder2019GPs,Rose2023Learning}, where closed-loop guarantees on the approximate MPC are provided using a posteriori sampling, similar to~\cite{Hertneck2018Learning, Nubert2020Robot}. 
These nonparametric estimation methods can in principle achieve any desired approximation accuracy $\error$,  however, they are in general not computationally cheap-to-evaluate. 
In particular, the complexity of evaluating kernel-based approximations scales cubicly in the number of samples~\cite[Section~2.3]{Rasmussen2006Gaussian}, making it
challenging to use
for control applications.  %
For kernel-based methods, computationally cheap online evaluation can, \eg be ensured by only considering nearest-neighbor samples for the online evaluation~\cite[Section~4.A]{Liu2020Scalable}
or 
by iteratively dividing the input space and using local GPs~\cite{Lederer2020Realtime, Nguyen2008Local}.

Overall, most state-of-the-art methods lack closed-loop guarantees, 
are limited to linear systems,
yield computationally expensive approximate controllers,
or require an iterative design process using a posteriori validation of error bounds.

\subsubsection*{Contribution}
In this work, we present \alkiax, a novel algorithm based on kernel interpolation~\cite{\kernelinterpolationcite},
a noise-free special case of KRR,
to automatically approximate functions up to any desired accuracy.
Hence, \alkiax\ enables the automatic approximation of nonlinear MPC schemes with closed-loop guarantees (see Figure~\ref{fig:framework}).
\alkiax\ has the following properties:
\begin{enumerate}[label=(P\arabic*)]
\item Fast-to-evaluate approximating function; \label{property:fast}
\item Guaranteed satisfaction of desired approximation error; \label{property:error_satisfaction}
\item Bound on worst-case number of required samples; 
\label{property:error_samples}
\item Automatic and non-iterative algorithm with well-conditioned computations and ready-to-use \python\ implementation.\footnote{%
The code can be found here:
\url{https://github.com/tokmaka1/ALKIA-X}}
\label{property:numerical}
\end{enumerate}
Property~\ref{property:fast} enables the implementation of the approximate MPC for fast control applications with limited computational resources. 
Property~\ref{property:error_satisfaction} ensures closed-loop guarantees on stability and constraint satisfaction on the approximate MPC and is the main theoretical contribution in combination with~\ref{property:error_samples}.
Properties~\ref{property:error_samples} and~\ref{property:numerical} enable a reliable and automatic approximation.
The aforementioned properties are obtained by developing, extending, and combining the following tools:
\begin{enumerate}[label=(T\arabic*)]
    \item Localized kernel interpolation; \label{tool:localized}
    \item Adaptive sub-domain partitioning and length scale adjustment; \label{tool:adaptive}
    \item RKHS norm extrapolation. \label{tool:RKHS}
\end{enumerate}
Specifically, we follow a localized kernel interpolation approach~\ref{tool:localized} by only using a subset of data, resulting in a piecewise-defined approximating function.
Additionally, we adaptively partition the domain into sub-domains~\ref{tool:adaptive}, adjust the length scale accordingly, and sample equidistantly in each sub-domain.
Finally, we derive a heuristic RKHS norm extrapolation~\ref{tool:RKHS}, eliminating the need for an oracle to provide the RKHS norm.

Overall, \alkiax\ successfully addresses the problem of approximating nonlinear MPC schemes with closed-loop guarantees by yielding a fast-to-evaluate approximating function~\ref{property:fast} and guaranteeing any desired approximation accuracy~\ref{property:error_satisfaction} with a reliable and automatic algorithm~\ref{property:error_samples},~\ref{property:numerical}.
We demonstrate the performance of \alkiax\ by approximating the MPC schemes for two nonlinear systems:
\begin{enumerate}
    \item A simple continuous stirred tank reactor from prior work~\cite{Hertneck2018Learning}. 
     \alkiax\ computes the approximate MPC in~\SI{9}{\hour}, yielding a controller with closed-loop guarantees and an online evaluation time below \SI{50}{\micro\second}, thus outperforming prior work~\cite{Hertneck2018Learning} by over one order of magnitude.
    \item A realistic application to control a cold atmospheric plasma device~\cite[Chapter~4.5]    {Bonzanini2022Thesis}.
    \alkiax\ computes the approximate MPC in~\SI{66}{\hour}, yielding a controller with only~\SI{33}{\mega\byte} of memory and an online evaluation time of \SI{100}{\micro\second}.
\end{enumerate}

\subsubsection*{Outline}
We present the problem setting in Section~\ref{sec:problem_setting}, where we formally reduce the MPC approximation problem into a function approximation problem. 
We use kernel interpolation to tackle the function approximation problem, for which we introduce preliminaries in Section~\ref{sec:kernel_interpolation}. 
Section~\ref{sec:alki} introduces the adaptive and localized kernel interpolation algorithm, including theoretical guarantees regarding the desired approximation error and worst-case sample complexity.
Section~\ref{sec:alkiax} extends the approach to unknown RKHS norms by using an RKHS norm extrapolation, yielding \alkiax.
Furthermore, Section~\ref{sec:alkiax_MPC} investigates approximating MPC schemes via \alkiax\ and the resulting closed-loop guarantees on the approximate MPC.
Section~\ref{sec:numerical} demonstrates the successful deployment of \alkiax\ for approximating two nonlinear MPC schemes. 
Finally, we conclude the paper in Section~\ref{sec:conclusion}.

\subsubsection*{Notation}
The set of non-negative real numbers is given by  $\mathbb{R}_{\geq 0}$. 
The set of positive real numbers 
and the set of natural numbers greater or equal to $N\in\mathbb{R}_{\geq 0}$
are denoted by $\mathbb{R}_{> 0}$ and $\mathbb{N}_{\geq N}$, respectively.
For a set 
$X \subseteq \mathbb{R}^N$,
we denote the cardinality by~$\mathrm{card}(X)$.
For a vector $x\in\mathbb{R}^n$, the Euclidean norm, the infinity norm, and the $1$-norm are denoted by
$\|x\|_2$, $\|x\|_\infty$, and $\|x\|_1$, respectively.
For a matrix $A\in\mathbb{R}^{n\times n}$, the induced infinity norm is denoted by $\|A\|_\infty$.
We define $\mathbf{1}_{n}=[1, \ldots, 1]^\top \in \mathbb{R}^{n}$ and $\mathbf{0}_{n}=[0, \ldots, 0]^\top \in \mathbb{R}^{n}$.
For two  functions $f{:}\; \mathbb{R}_{\geq 0}\rightarrow \mathbb{R}_{\geq0}$ and $g{:}\; \mathbb{R}_{\geq 0}\rightarrow \mathbb{R}_{\geq0}$, we write $f(\error) = \mathcal{O}(g(\error))$ 
if $\frac{f(\epsilon)}{g(\epsilon)} < \infty$
as $\epsilon \rightarrow 0$ with $\error >0$.
Moreover, a function $f{:}\; \mathbb{R}_{\geq 0} \rightarrow \mathbb{R}_{\geq 0}$ belongs to class $\mathcal{K}_\infty$ if it is strictly increasing, continuous, and $f(0)=0$.


%% file: figures_final/figure_intro.tex
\begin{tikzpicture}
     \node[draw, align=left, minimum width=8.5cm] at (0,0) (A) {
     \begin{minipage}{8.5cm}
     \centering
         \textbf{\underline {Robust MPC}}
     \end{minipage}
      \\~\\
     \begin{minipage}{8.5cm}
         \begin{itemize}[left=0em]
             \item          Closed-loop guarantees for disturbances on control input bounded by~$\error$
         \end{itemize}
     \end{minipage}
     };
    \node[draw, align=left, fill=gray!30, minimum width=8.5cm, below=1cm of A.south] (B) {
     \begin{minipage}{8.5cm}
         \centering
         \textbf{\underline{\alkiax}}
     \end{minipage}
     \\~\\
    \begin{minipage}{8.5cm}
        \begin{itemize}[left=0em]
            \item Automatic computation of approximate MPC
            \item Uniform approximation error bounded by $\error$ (Theorem~\ref{th:alkiax}) 
        \end{itemize}
    \end{minipage}
     };
     \node[draw, align=left, minimum width=8.5cm, below=1cm of B.south] (C) {
     \begin{minipage}{8.5cm}
     \centering
     \textbf{\underline{Approximate MPC}}
     \end{minipage}
     \\~\\
        \begin{minipage}{8.5cm}
            \begin{itemize}[left=0em]
                \item Fast-to-evaluate controller
                \item Closed-loop guarantees (Corollary~\ref{co:MPC})
            \end{itemize}
        \end{minipage}
     };
     \draw[-latex] ([yshift=1cm]A.north) -- (A.north);
     \draw[-latex] (A.south) -- (B.north);
     \draw[-latex] (B.south) -- (C.north);
\node[align=left, text width=5cm] at  ($(A.north) + (2.75,0.5)$) {User-defined bound~$\error$ \\ on approximation error};
 \node[align=left, text width=5cm] at ($(B.north) + (2.75,0.5)$) {Samples of MPC};
     \node[align=left, text width=5cm] at ($(C.north) + (2.75,0.5)$) {Approximating function};
\end{tikzpicture}

%% file: Sections/preliminaries.tex
\section{Problem Setting and Preliminaries}
\label{sec:p_and_p}
First, we describe the problem setting (Section~\ref{sec:problem_setting})
before
giving preliminaries on kernel interpolation (Section~\ref{sec:kernel_interpolation}).
Then, we discuss the computation of an upper bound on the power function (Section~\ref{sec:power_function}).

\subsection{Problem setting}\label{sec:problem_setting}
We consider a general nonlinear dynamical system 
$
    x_{t+1}=g(x_t,u_t)
$
with system state $x_t\in\mathbb{R}^n$, control input $u_t\in\mathbb{R}^{n_{\mathrm{u}}}$, and time index $t\in\mathbb{N}$. The control goal is to stabilize the steady-state~$x_{\mathrm{s}}$ while ensuring the satisfaction of user-specified state and input constraints, i.e., $x_t\in\domain \subseteq \mathbb{R}^n$, $u_t\in\mathcal{U}\subseteq \mathbb{R}^{n_{\mathrm{u}}}$, $\forall t\in\mathbb{N}$.
MPC provides a constructive approach to define a feedback law $u=f(x)$  that ensures asymptotic stability and constraint satisfaction.
\subsubsection*{Approximate MPC}
The feedback law~$f$ is based on the solution of a finite-horizon optimal control problem and solving such an NLP during online operation is computationally expensive.
Hence, we aim to replace $f$ with an explicit approximate feedback law $u = h(x)$ that we construct offline.
The main challenge is to ensure that the approximate MPC~$h$ inherits the closed-loop guarantees induced by the MPC~$f$. 
\subsubsection*{Solution approach}
The proposed framework is outlined in Figure~\ref{fig:framework}.
Analogous to the approach in~\cite{Hertneck2018Learning, Nubert2020Robot}, we consider an MPC feedback law~$f$ that is robust with respect to input disturbances bounded by~$\error$.
Hence, the approximate MPC yields desirable closed-loop guarantees if $\|f(x)-h(x)\|_\infty \leq \error$ for all $x\in\domain$. 

\subsubsection*{Function approximation}
Consequently, the problem of determining an approximate MPC with closed-loop guarantees reduces to computing an explicit function $h$ that approximates the ground truth $f$ with a uniform approximation error bounded by~$\error$, \ie
$\|f(x)-h(x)\|_\infty \leq \error$ for all $x \in \domain$.
The ground truth~$f$ is implicitly defined through the solution of an NLP, and hence we regard~$f$ as a black-box function that we can query to receive noise-free evaluations.
In the following, we focus solely on this general function approximation problem. 
Specifically, in this paper, we solve this function approximation problem by proposing a novel algorithm based on kernel interpolation. 
We revisit the resulting control properties of the approximate MPC scheme later in Section~\ref{sec:alkiax_MPC}.
\subsubsection*{Simplifications}
We consider, w.l.o.g., a \emph{scalar} ground truth $f{:}\;\domain\rightarrow \mathbb R$ and approximating function $h{:}\;\domain\rightarrow \mathbb R$, and thus
require
\begin{equation} \label{eq:uniform_error}
\lvert f(x) - h(x) \rvert  \leq \error \quad \forall x \in \domain.
\end{equation}
For vector-valued functions, each output dimension can be approximated individually, resulting in multiple scalar approximation problems ensuring $\|f(x)-h(x)\|_\infty \leq \error$ for all $x\in\domain$.
For simplicity of exposition, we assume that the domain is a unit cube, \ie $\domain=[0,1]^n$.

\subsection{Preliminaries on kernel interpolation}
\label{sec:kernel_interpolation}
In this section, we introduce kernel interpolation, see~\cite{\kernelinterpolationcite} for more details.
First, we collect the standing assumptions on the kernel $k{:} \; \domain {\times} \domain \rightarrow \mathbb{R}_{> 0}$, which is the central object in kernel interpolation.
\begin{assumption}\label{asm:kernel}
The kernel~$k$ is
\begin{enumerate}
\item continuous.
\item (strictly) positive definite, \ie $\sum_{i=1}^N\sum_{j=1}^N \alpha_i \alpha_j k(x_i,x_j) > 0$ for all $\mathbb{R}^N \ni[\alpha_1,\ldots, \alpha_N]^\top \not = \mathbf{0}_N$, for all pairwise distinct $x_1, \ldots, x_N \in \domain$, and any $N \in \mathbb{N}$.
\item only a function of the Euclidean distance, \ie $\exists \widetilde{k}{:}\; \mathbb{R}_{\geq 0} \rightarrow (0,1]$ such that $k(x,x^\prime) = \widetilde{k}(\|x-x^\prime\|_2)$ for all $x,x^\prime \in \domain$.
Moreover, $\widetilde k$ is strictly decreasing and $\widetilde{k}(0) = 1$.
\end{enumerate}
\end{assumption}

\noindent
The first two conditions in Assumption~\ref{asm:kernel} are standard assumptions for kernels~\cite{Maddalena2021Deterministic}.
Condition~$3)$ implies a uniform length scale in all dimensions and that $k$ is a radial kernel~\cite[Chapter~4.7]{Steinwart2008SVM}.
We can assume $\widetilde k(0)=1$ w.l.o.g. by normalizing.
Assumption~\ref{asm:kernel} holds for many common kernels, such as the squared-exponential or the Matérn kernel~\cite{Rasmussen2006Gaussian}. 

Let $X=\{x_1, \ldots, x_N\} \subseteq \domain$ be a set of samples and assume that~$x_1, \ldots, x_N$ are pairwise distinct.
We collect the noise-free evaluations of $f$ on the inputs $X$ in a vector denoted by
\begin{align} \label{eq:f}
f_X = [f(x_1), \ldots, f(x_N)]^\top. 
\end{align}
Given samples $X$, we denote by  $K_X \in \mathbb{R}_{> 0}^{N \times N}$ the symmetric and positive definite covariance matrix, also known as the kernel matrix or the Gram matrix~\cite{Steinwart2008SVM, Rasmussen2006Gaussian}, \ie the entry at row~$i$ and column~$j$ is $k(x_i,x_j)$ for any $i,j \in \{1, \ldots N\}$.
Given samples $X$,
the covariance vector $k_X{:} \; \domain \rightarrow \mathbb{R}_{> 0}^{N}$ is defined as
\begin{align} \label{eq:k_x}
k_X(x) \coloneqq \left[k(x,x_1), k(x,x_2), \ldots, k(x,x_N)\right]^\top.
\end{align}
With the introduced definitions, we can write the resulting approximating function $h_X{:}\; \domain \rightarrow \mathbb{R}$ as~\cite[Section~3.2]{Kanagawa2018Gaussian}
\begin{align}\label{eq:h}
h_X(x) = f_X^\top K_X^{-1} k_X(x).
\end{align}
The RKHS norm is a characteristic value of RKHS functions that quantifies their complexity, and which we later use to bound the approximation error.
We denote the RKHS norm of a function with respect to the RKHS of kernel $k$ by $\|\cdot\|_k$ and the RKHS norm of the approximating function~\eqref{eq:h} satisfies
\begin{align} \label{eq:h_RKHS}
\|h_{X}\|_k = \sqrt{f_{X}^\top K_X^{-1} f_X}.
\end{align}
Notably, approximating function~\eqref{eq:h} is the unique solution of a variational problem~\cite[Theorem~3.5]{Kanagawa2018Gaussian}, interpolating the given samples, \ie $h_X(x) = f(x)$ for all $x \in X$, with the minimal RKHS norm. 
Thus, $\|f\|_k \geq \|h_X\|_k$ holds for all $X \subseteq \domain$.
Henceforth, we assume that the unknown ground truth $f$ is a member of the RKHS of the chosen kernel $k$, which trivially implies $\|f\|_k < \infty$. 
This is a standard assumption for kernel-based approximation (see \eg\cite{\RKHSnormcite}), which we discuss in more detail later (see Assumptions~\ref{asm:oracle_RKHS_adaptive} and~\ref{asm:extrapolation}).

Next, we define the power function $P_X{:} \;\domain \rightarrow [0,1)$~\cite[Section~9.3]{Fasshauer2015Kernel} 
\begin{align}
P_X(x) \coloneqq \sqrt{1-k_X(x)^\top K_X^{-1} k_X(x)} \label{eq:power},
\end{align}
which is used in the following proposition to characterize input-dependent error bounds for kernel interpolation.
%
\begin{proposition} \label{le:maddalena}
\cite[Section~14.1]{Fasshauer2015Kernel}
Let Assumption~\ref{asm:kernel} hold.
Then, we have
\begin{align} \label{eq:error_lemma}
\lvert f(x) - h_X(x) \rvert \leq P_X(x) \sqrt{\|f\|_k^2 - \|h_X\|_k^2}
\end{align}
for all $x \in \domain$ and any samples $X \subseteq \domain$.
\end{proposition}
Note that Proposition~\ref{le:maddalena} is a modified version of the Golomb-Weinberger bound~\cite[Section~9.3]{Fasshauer2015Kernel}.
Based on Proposition~\ref{le:maddalena}, we can obtain \emph{uniform} bounds on the approximation error given a uniform upper bound on the power function~$P_X$ and the ground truth RKHS norm~$\|f\|_k$.
Additionally, the following lemma ensures continuity of the RKHS function.
\begin{lemma}\label{le:cont}
Let Assumption~\ref{asm:kernel} hold.
Then, for all $x, x^\prime \in \domain$, we have
\begin{align}
    \lvert f(x) - f(x^\prime) \rvert \leq \|f\|_k \sqrt{2\knorm(\|x-x^\prime\|_2)}
\end{align}
with
\begin{align} \label{eq:knorm}
     \knorm \coloneqq 1-\widetilde k \in \mathcal K_\infty.\footnotemark[2]
\end{align}
\end{lemma}
The proof can be found in Appendix~\ref{app:cont}.
Lemma~\ref{le:cont} later assists in proving sample complexity bounds for arbitrarily accurate approximating functions.
\footnotetext[2]{Due to Assumption~\ref{asm:kernel}, the function $\knorm{:}\; \mathbb{R}_{\geq 0} \rightarrow [0,1)$ belongs to class $\mathcal{K}_\infty$,
and thus its inverse $\kinv{:}\; [0,1) \rightarrow \mathbb{R}_{\geq 0}$ is a continuous and strictly increasing function, \ie it is again a class $\mathcal{K}_\infty$ function~\cite[Lemma~4.2]{Khalil2002Nonlinear}.}
\addtocounter{footnote}{1} 

\subsection{Upper bound on the power function}\label{sec:power_function}
Based on~\eqref{eq:error_lemma}, we can ensure the desired uniform error bound~\eqref{eq:uniform_error} if a bound on the RKHS norm of the ground truth and a uniform upper bound on the power function $P_X$ are available.
However, finding the maximum of $P_X$ over the space $\domain$ in general requires the solution of a non-convex optimization problem.
We circumvent this issue by using a local approximation scheme with samples $X$ at the $2^n$ vertices of a cube.
In this symmetric setting, the center of the cube is the farthest away from the given samples.
Since \emph{radial} kernels (see  Assumption~\ref{asm:kernel}) are strictly decreasing functions of the Euclidean distance between two inputs, it is reasonable to assume that the maximum of the power function is at the center of the cube.

\begin{assumption}\label{asm:max_power}
For any cube~$\widetilde\domain\subseteq\domain$ with $2^n$ samples~$X$ located at its vertices and with the center $\overline{x}=\frac{1}{2^n}\sum_{x \in X} x$, it holds that $P_{X}(x) \leq P_{X}(\overline{x})$ for all $x \in \widetilde\domain$.
\end{assumption}

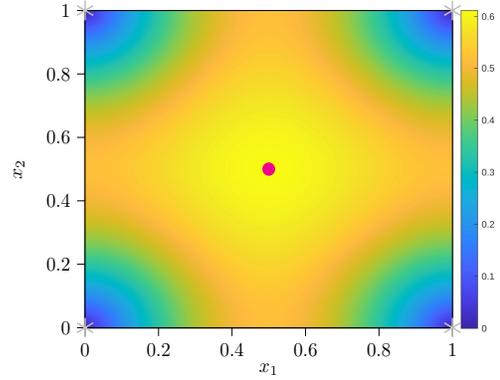
\begin{figure}
    \centering
    \input{figures_final/power_function_figure.tex}
    \caption{
  Power function $P_X$ on a cube with samples $X$ on its vertices, which are illustrated by the gray asterisks. The maximum of the power function is at the center and is highlighted by the magenta point. 
  This plot is generated with the Matérn kernel with $\nu=\nicefrac{3}{2}$ to compute the power function on the cube~$[0,1]^2$.
}
    \label{fig:quasiconcave}
\end{figure}
\noindent
Figure~\ref{fig:quasiconcave} shows an exemplary power function, highlighting that the maximum of the power function is indeed at the center of the cube.

%% file: figures_final/power_function_figure.tex
\begin{tikzpicture}[xscale=.7125, yscale=.74125]
\centering
\definecolor{darkgray176}{RGB}{176,176,176}
\pgfplotsset{
every axis legend/.append style={
at={(0.5,0.5)},
anchor=north west,
},
}

\node[inner sep=0pt] (whitehead) at (3.83, 2.8125)
    {\includegraphics[width=5.5cm]{figures_final/power_function_2D_cropped_downsized.jpg}};
\begin{axis}[
tick align=outside,
tick pos=left,
xlabel=$x_1$,
xmin=0, xmax=1,
xtick style={color=black},
ylabel={$x_2$},
ymin=0,
ymax=1,
ytick style={color=black}
]

\addplot [thick, lightgray, mark=asterisk, mark size=5, mark options={solid}, only marks]
table {%
0 0
0 1
1 0
1 1
};

\addplot [thick, magenta, mark=*, mark size=3, mark options={solid}, only marks]
table {%
0.5 0.5
};

\end{axis}

\end{tikzpicture}

%% file: Sections/alki.tex
\section{Adaptive and localized kernel interpolation algorithm}
\label{sec:alki}
In this section, we present the proposed adaptive and localized kernel interpolation algorithm to automatically compute a sufficiently accurate approximating function.
In Section~\ref{sec:alki_main_idea}, we describe the main idea before explaining the details in Section~\ref{sec:alki_detailed}.
We provide a theoretical analysis in Section~\ref{sec:alki_theory}, including the guaranteed error bound and a bound on the worst-case sample complexity. 

\subsection{Main idea} 
\label{sec:alki_main_idea}
A naïve approach to determine a sufficiently accurate approximating function would be to equidistantly sample the ground truth until the kernel interpolation ensures the desired error bound~\eqref{eq:uniform_error}.
The proposed approach modifies classic kernel interpolation by using \emph{local} approximations and an \emph{adaptive} adjustment of the length scale.

We use a \emph{localized} kernel interpolation approach~\ref{tool:localized} by locally sampling equidistantly and computing a piecewise-defined approximating function based on local cubes that only use a subset of the available samples.
This results in a fast online evaluation~\ref{property:fast}, 
a highly parallelized offline approximation,
and enables the reliable computation of a guaranteed bound on the approximation error~\ref{property:error_satisfaction},~\ref{property:error_samples}.

Additionally, we use an \emph{adaptive} adjustment of the length scale by partitioning the domain into sub-domains~\ref{tool:adaptive} and perform localized kernel interpolation on each sub-domain separately, where each sub-domain has an individual length scale and RKHS norm.
For every sub-domain, we can compute the sufficient number of samples that ensure an approximation error that is uniformly  bounded by $\error$. 
Moreover, we impose an upper bound on the number of samples per sub-domain using a hyperparameter $\overline p \in \mathbb{N}$.
If more samples are required to ensure the desired error bound, we partition the sub-domain and reduce the length scale.
This leads to an \emph{adaptive} sub-domain partitioning with denser sampling in harder-to-approximate areas.
Furthermore, upper-bounding the number of samples for each sub-domain enforces an upper bound on the condition number of the covariance matrices, leading to numerically reliable computations, even for millions of samples, without requiring any regularization~\ref{property:numerical}.

\subsection{Proposed algorithm} \label{sec:alki_detailed}
We partition the domain $\domain$ into sub-domains $\domain_a$, \ie$\domain = \cup_{a \in \mathcal{A}} \domain_a$, $\mathcal{A} \subseteq \mathbb{N}$, where each sub-domain is a cube.
Each sub-domain $\domain_a$ has an individual length scale $\ell_a$ and an individually scaled kernel $\kla$ with
\begin{align}\label{eq:scaled_kernel}
\kla(\cdot) \coloneqq \widetilde{k} \left( \frac{\cdot}{\ell_a} \right), \quad
\ell_a = \underset{x_i,x_j \in \domain_a}{\mathrm{max}}\|x_i-x_j\|_\infty.
\end{align}
Accordingly, the covariance matrix, the covariance vector, and the function $\knorm$ (see~\eqref{eq:knorm}) for each sub-domain $a \in \mathcal A$ are denoted by $K_{a,X}$, $k_{a,X}$, and $\knorm_{a}$, respectively.

We sample equidistantly in every sub-domain $\domain_a$ with grid size $\Delta x_a > 0$, which is chosen such that $\nicefrac{\ell_a}{\Delta x_a}=2^{p_a}$ with some later specified  $p_a \in \mathbb{N}_{\geq \underline p}$
and a  user-chosen hyperparameter~$\underline{p} > 0$.
For any $p \geq \underline p$, we denote by $X_{a,p}$ the set of $(1+2^p)^n$ equidistant samples on sub-domain $\domain_a$.
The total samples on the domain $\domain$ are thus given by $X=\cup_{a\in\mathcal{A}} X_{a,p_a}$. 

Additionally, we divide the sub-domain $\domain_a$ into uniform local cubes $\domain_c$, \ie$\domain_a = \cup_{c \in \mathcal{C}_{a,p_a}} \domain_c$, $\cubes_{a,p_a} \subseteq \mathbb{N}$, of edge length $\Delta x_a$.
Each local cube $\domain_c$ has $2^n$ samples on its vertices, which we denote by $X_c$.
The center $\xcm$ of local cube $\domain_c$ is given by $\xcm=\frac{1}{2^n}\sum_{x \in X_c} x$.
The samples on sub-domain~$\domain_a \supseteq \domain_c$ can be obtained with $X_{a,p_a}=\cup_{c \in \cubes_{a,p_a}} X_c$.

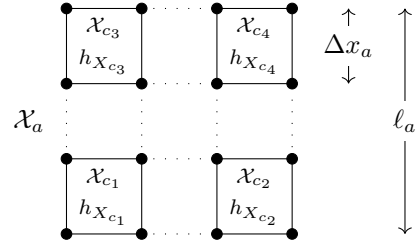
\begin{figure}
\centering
\begin{tikzpicture}

   \foreach \point [count=\i] in {
   (0,0), (1,0), (2,0), (3,0),
   (0,1),(1,1),(2,1),(3,1),
   (0,2), (1,2), (2,2), (3,2),
    (0,3), (1,3), (2,3), (3,3)
    } {
     \node[coordinate] (h\i) at \point { };
   }

    \filldraw (h1) circle (2pt);
    \filldraw (h2) circle (2pt);
    \filldraw (h3) circle (2pt);
    \filldraw (h4) circle (2pt);

    \filldraw (h5) circle (2pt);
    \filldraw (h6) circle (2pt);
    \filldraw (h7) circle (2pt);
    \filldraw (h8) circle (2pt);

    \filldraw (h9) circle (2pt);
    \filldraw (h10) circle (2pt);
    \filldraw (h11) circle (2pt);
    \filldraw (h12) circle (2pt);

    \filldraw (h13) circle (2pt);
    \filldraw (h14) circle (2pt);
    \filldraw (h15) circle (2pt);
    \filldraw (h16) circle (2pt);

    \node[align=left] at (-0.5, 1.5){{$\domain_a$}};
    \draw[<->] (4.5,3) --  (4.5,0) node[midway,fill=white] {{$\ell_a$}};
    \draw[<->] (3.75,3) --  (3.75,2) node[midway, fill=white] {{$\Delta x_a$}};

    \draw (h1)--(h2);
    \draw (h1)--(h5);
    \draw (h5)--(h6);
    \draw (h2)--(h6);

    \draw (h3)--(h4);
    \draw (h3)--(h7);
    \draw (h7)--(h8);
    \draw (h4)--(h8);

    \draw (h9)--(h10);
    \draw (h9)--(h13);
    \draw (h13)--(h14);
    \draw (h10)--(h14);

    \draw (h11)--(h12);
    \draw (h11)--(h15);
    \draw (h15)--(h16);
    \draw (h12)--(h16);

    \draw[loosely dotted] (h2)--(h3);
    \draw[loosely dotted] (h6)--(h7);
    \draw[loosely dotted] (h10)--(h11);
    \draw[loosely dotted] (h14)--(h15);

    \draw[loosely dotted] (h5)--(h9);
    \draw[loosely dotted] (h6)--(h10);
    \draw[loosely dotted] (h7)--(h11);
    \draw[loosely dotted] (h8)--(h12);
    \node[align=center] at (0.5,0.5){\footnotesize{$\domain_{c_1}$} \\ \footnotesize{$h_{X_{c_1}}$}};
    \node[align=center] at (2.5,0.5){\footnotesize{$\domain_{c_2}$} \\ \footnotesize{$h_{X_{c_2}}$}};
    \node[align=center] at (0.5,2.5){\footnotesize{$\domain_{c_3}$} \\ \footnotesize{$h_{X_{c_3}}$}};
    \node[align=center] at (2.5,2.5){\footnotesize{$\domain_{c_4}$} \\ \footnotesize{$h_{X_{c_4}}$}};
\end{tikzpicture}
\caption{Structure of the resulting uniform local cubes and the localized approximating functions on a sub-domain $\domain_a$ with length scale $\ell_a$ and grid size $\Delta x_a$.
Illustrated are the local cubes $\domain_c$, $c \in \{c_1,c_2,c_3,c_4\} \subseteq \mathcal{C}_{a,p_a}$.
}
\label{fig:cubes}
\end{figure}
The resulting structure is shown in Figure~\ref{fig:cubes}.
For each local cube $\domain_c$, we define a localized approximating function~$h_{X_c}$ that only depends on the $2^n$ samples~$X_c$, leading to a piecewise-defined approximating function and a localized kernel interpolation approach.

For the approximation procedure, we use function values that are shifted by an empirical mean value, \ie
\begin{align}\label{eq:shift}
    \widetilde{f}_{X_{a,p}} = f_{X_{a,p}} - \mu_a, \quad
    \mu_a \coloneqq \frac{\sum_{x \in {X_{a,\underline{p}}}}f(x)}{\mathrm{card}(X_{a,\underline p})}.
\end{align}
Due to the continuity of the ground truth (see  Lemma~\ref{le:cont}), this shifting ensures arbitrarily small function values $\widetilde f_{X_{a,p}}$ for decreasing size of the sub-domain $\domain_a$, which later assists in guaranteeing an arbitrarily small approximation error $\error > 0$.
For any local cube $\domain_c \subseteq \domain_a$ with samples $X_c \subseteq X_{a,p_a}$, the localized kernel interpolation with shifted function values~\eqref{eq:shift} results in the shifted localized approximating function (see~\eqref{eq:h})
\begin{align}\label{eq:shifted_h_cube}
    \widetilde {h}_{X_{c}}(x) = \widetilde f_{X_c}^\top K_{a, X_c}^{-1} \widetilde k_{a, X_c}(x).
\end{align}

To construct the shifted localized approximating functions on a sub-domain $\domain_a$, we first query the ground truth at the inputs $X_{a,p_a}$ and shift the samples by the empirical mean~$\mu_a$.
With these samples, we compute the shifted localized approximating functions~\eqref{eq:shifted_h_cube} for all local cubes $\domain_c \subseteq \domain_a$.
Algorithm~\ref{alg:alki_function}  summarizes this procedure.

\begin{algorithm}
\begin{algorithmic}[1]
\Require $f$, $\kla$, $X_{a,p_a}$
    \State Sample $f_{X_{a,p_a}}$
 \Comment{\eqref{eq:f}}
    \State Compute mean $\mu_a$ and shifted values $\widetilde f_{X_{a,p_a}}$  \Comment{\eqref{eq:shift}} 
    \State Define cube partition $\mathcal{C}_{a,p_a}$ from $X_{a,p_a}$ 
\For {$c \in \cubes_{a,p_a}$} 
	\State Build approximating function $\widetilde h_{X_c}$
 \Comment{\eqref{eq:shifted_h_cube}}
\EndFor
\end{algorithmic}
\caption{Shifted localized approximating functions}
\label{alg:alki_function}
\end{algorithm}

In the following, we ensure that the shifted localized approximating functions computed in Algorithm~\ref{alg:alki_function} satisfy the desired error bound~$\error$.
First, we define the shifted ground truth $\widetilde f_a$ on sub-domain $\domain_a$ as
\begin{align}\label{eq:shifted_f}
    \widetilde{f}_{a}(x) \coloneqq f(x) - \mu_a.
\end{align}
Given the sub-domain partitioning and the local approximation scheme, the desired approximation error bound~\eqref{eq:uniform_error} reduces to
\begin{align} \label{eq:uniform_error_cube}
\lvert \widetilde f_a(x) - \widetilde h_{X_{c}}(x)\rvert \leq \error \quad \forall a \in \mathcal{A},  \; \forall c \in \cubes_{a,p_a}, \;    \forall x \in \domain_c.
\end{align}
Furthermore, analogous to standard kernel-based approximation methods
(see \eg\cite{\RKHSnormcite}), 
we require knowledge of an upper bound on the corresponding RKHS norm.
\begin{assumption}\label{asm:oracle_RKHS_adaptive}
For any $a\in\mathcal{A}$, $\widetilde f_a$ is a member of the RHKS of kernel $\kla$. 
Moreover, we have access to an oracle $\rkhso \in (0, \infty)$, which is an upper bound on the RHKS norm, \ie $\rkhso \geq \|\widetilde f_a\|_\kla$.
\end{assumption}
\noindent
In Section~\ref{sec:alkiax}, we introduce a heuristic approximation that replaces the oracle assumption.

The following result provides a lower bound on the number of samples that ensure the desired error bound $\error$ for each sub-domain $\domain_a$.

\begin{proposition} \label{pr:max_error}
Let Assumptions~\ref{asm:kernel}--\ref{asm:oracle_RKHS_adaptive} hold.
For any $a \in \mathcal{A}$ and any~$\error > 0 $, there exists a $p_a^* \in \mathbb{N}$ 
such that $p_a \geq p_a^*$ implies $P_{X_{c}}(\xcm) \rkhso \leq \error$ for all  $c \in \cubes_{a,{p_a}}$.
Furthermore, if $p_a \geq p_a^*$ for all $a \in \mathcal{A}$, then the localized approximating functions according to Algorithm~\ref{alg:alki_function} satisfy~\eqref{eq:uniform_error_cube}.
\end{proposition}

\begin{proof}
The proof is divided into three parts.
In Part~I, we bound the power function~$P_{X_c}(\xcm)$ using an eigenvector of the covariance matrix.
In Part~II, we show that $P_{X_c}(\xcm)\rkhso \leq \error$ holds for $p_a \geq p_a^*$, while in Part~III we prove that this implies~\eqref{eq:uniform_error_cube}.\\
\textit{Part~I:}
Pick any $a \in \mathcal{A}$.
For any $c \in \cubes_{a,p_a}$,
given samples~$X_c$ on the vertices of local cube~$\domain_c$,
we denote the scaled covariance matrix and the scaled covariance vector by~$K_{a,X_c}$ and~$k_{a,X_c}$, respectively.
In this cubic setting, the covariance matrix~$K_{a,X_c}$ has constant row sum.
Moreover, the covariance vector satisfies~$k_{a,X_c}(\xcm) = \widetilde{k}_{a}\left(\frac{\sqrt{n} \Delta x_a}{2}\right)\cdot \mathbf{1}_{2^n}$.
Therefore,~$k_{a,X_c}(\xcm)$ is an eigenvector of $K_{a,X_c}$ with eigenvalue 
\begin{align}\label{eq:beta}
\beta = \sum_{x \in X_c} \widetilde{k}_{a}(\|x-x^\prime\|_2) \overset{\text{Asm.\ref{asm:kernel}}}{\leq}  2^n    
\end{align}
for an arbitrary $x^\prime \in X_c$.
Hence, we have
\begin{align}\label{eq:eigenvector}
    K_{a,X_c}^{-1} k_{a,X_c}(\xcm) = \frac{1}{\beta} k_{a,X_c}(\xcm).
\end{align}
This yields
\begin{align} \label{eq:P_eigen}
P_{X_c}(\xcm) \overset{\eqref{eq:power},\eqref{eq:eigenvector}}&{=}
\sqrt{1-
\frac{\widetilde{k}_{a}\left(\frac{\sqrt{n}\Delta x_a}{2}\right)^2 2^n}
{\beta} 
}
\\
\overset{\eqref{eq:beta}}&{\leq}
\sqrt{1-
\widetilde{k}_{a}\left(\frac{\sqrt{n}\Delta x_a}{2}\right)^2
}. \nonumber
\end{align}
\textit{Part~II}:
We introduce the notation $\lceil \xi \rceil \in \mathbb{N}$ as rounding $\xi \in \mathbb{R}$ up to the nearest integer.
First, suppose $\error < \rkhso$ and define
\begin{align}\label{eq:p_log}
    p_a^* := \left\lceil \log_2\left(\frac{\ell_a \sqrt{n}}{2\kinv_{a}\left(\frac{1}{2}\left(\frac{\error}{\rkhso}\right)^2\right)}\right) \right\rceil.
\end{align}
For any $p_a \geq p_a^*$, any $c \in \cubes_{a,p_a}$, and any~$\error > 0$, it holds that
\begin{align} \label{eq:p_bound}
        &P_{X_c}(\xcm) \rkhso \overset{\eqref{eq:P_eigen}}{\leq}\rkhso \sqrt{1-
\widetilde{k}_{a}\left(\frac{\sqrt{n}\Delta x_a}{2}\right)^2
}  
\nonumber
\\&=
\rkhso \sqrt{ \left(1-
\widetilde{k}_{a}\left(\frac{\sqrt{n}\Delta x_a}{2}\right)\right)
\left(1+
\widetilde{k}_{a}\left(\frac{\sqrt{n}\Delta x_a}{2}\right)
\right)
} \nonumber \\ 
\overset{\eqref{eq:knorm}}&{=}\rkhso \sqrt{ {\knorm}_{a}\left(\frac{\sqrt{n}\Delta x_a}{2}\right)  
\left(1+
\widetilde{k}_{a}\left(\frac{\sqrt{n}\Delta x_a}{2}\right)
\right)
} \\
&\leq
\rkhso \sqrt{ 2{\knorm}_{a}\left(\frac{\sqrt{n}\Delta x_a}{2}\right) 
}
=
\rkhso \sqrt{ 2{\knorm}_{a}\left(\frac{\sqrt{n} \ell_a}{2 \cdot 2^{p_a}}\right) 
}
\nonumber \\
&\overset{\eqref{eq:p_log}}{\leq} 
\rkhso \sqrt{2\knorm_{a}\left(\kinv_{a}
    \left(\frac{1}{2}\left(\frac{\error}{\rkhso}\right)^2
    \right)
    \right)} = \error. \nonumber
\end{align}
Note that we used $2^{p_a} = \nicefrac{\ell_a}{\Delta x_a}$ and the fact that $\knorm_a \in \mathcal{K}_\infty$ and $\kla{:}\; \mathbb{R}_{\geq0} \rightarrow (0,1]$.
%
In case of $\error \geq \rkhso$, $P_{X_c}(\xcm) \rkhso \leq \error$ 
(see~\eqref{eq:p_bound}) holds for any $\Delta x_a>0$ since $P_{X_c}(\xcm) \leq 1$ for all $p_a\geq p_a^*=0$.\\
\textit{Part~III:}
For any $c \in \mathcal{C}_{a,p_a}$ and for any $x \in \domain_c$, given the sub-domain partitioning and the local approximation scheme,
error bound~\eqref{eq:error_lemma} of Proposition~\ref{le:maddalena} yields
\begin{align} \label{eq:error_lemma_cube}
\lvert \widetilde f_a(x) - \widetilde h_{X_{c}}(x) \rvert \leq P_{X_{c}}(x) \|\widetilde f_a\|_\kla.
\end{align}
Thus, for any $p_a \geq p_a^*$, any $c \in \cubes_{a,p_a}$, any $x \in \domain_c$, and any~$\error >0$, it holds that
\begin{align*}
\lvert \widetilde f_a(x) - \widetilde h_{X_{c}}(x) \rvert\overset{\eqref{eq:error_lemma_cube}}&{\leq} P_{X_{c}}(x) \|\widetilde f_a\|_\kla
\overset{\text{Asm.}\ref{asm:max_power}}{\leq} P_{X_{c}}(\xcm) \|\widetilde f_a\|_\kla \\
\overset{\text{Asm.}\ref{asm:oracle_RKHS_adaptive}}&{\leq} P_{X_{c}}(\xcm) \rkhso \leq
\error,
\end{align*}
where the last inequality was shown in Part~II.
\end{proof}

Proposition~\ref{pr:max_error} ensures arbitrary approximation accuracy $\error$ using a local approximation scheme by sampling densely enough.
Since we additionally enforce a lower bound $p_a\geq \underline{p}$, we choose 
\begin{align} \label{eq:pmin_pa}
    p_a = \max\{\underline p, p_a^*\}.
\end{align}
For small accuracies~$\error$, this may result in large values~$p_a$, which deteriorates the conditioning of the covariance matrices~$K_{a,X_c}$, resulting in numerical unreliability.
Thus, we introduce~$\overline p > \underline p$, an upper bound on the integer~$p_a$.
This yields a uniform upper bound~$\overline \kappa$ on the condition number
\begin{align} \label{eq:condition}
   \mathrm{cond}\left(K_{a,X_c}\right) \leq \overline{\kappa} \quad \forall a \in \mathcal{A}, \; \forall p \in \{\underline{p},\ldots,\overline{p}\}, \; \forall c \in \cubes_{a,p},
\end{align}
with $\mathrm{cond}(K_{a,X_{c}}) \coloneqq \|K_{a,X_{c}}\|_\infty \|K_{a,X_{c}} ^{-1}\|_\infty$.
For a given~$p_a \geq \underline p$, we have a fixed ratio $\nicefrac{\ell_a}{\Delta x_a} = 2^{p_a}$, and hence~$K_{a, X_c}$ is identical for any~$c \in C_{a,p_a}$ and any~$a \in \mathcal{A}$.
Thus, given a uniformly bounded~$p_a \in \{\underline p, \ldots, \overline p\}$, this yields a uniform bound on the condition number~$\overline \kappa$. 
Furthermore,~$\overline p$ and~$\overline \kappa$ have one-to-one correspondence.\footnote{%
In fact,~\eqref{eq:condition} can be enforced a priori.
First, we define $\mathcal{C}_{\overline p}$ as the cube partition resulting from $\mathrm{card}(X)=(1+2^{\overline p})^n$ equidistant samples on the domain $\domain = [0,1]^n$.
Then, we compute the covariance matrix of the unscaled kernel $k$ (\ie with $\ell=1$) for an arbitrary $c \in \mathcal{C}_{\overline p}$ with $2^n$ samples $X_c$ and check whether $\mathrm{cond}(K_{X_c}) \leq \overline\kappa$.}
If~\eqref{eq:pmin_pa} returns $p_a > \overline p$, we partition the sub-domain $\domain_a$ into $2^n$ cubes of the same size, which halves the length scale $\ell_a$, allowing for a denser sampling.

The overall offline approximation is summarized in Algorithm~\ref{alg:alki_offline}.
To compute a sufficiently accurate approximation of the ground truth~$f$ on the domain~$\domain$,
we iterate through a dynamic for-loop to determine the localized approximating functions for all sub-domains.
After querying the oracle to receive an upper bound on the RKHS norm of the ground truth~$\rkhso$, we determine the number of samples for the current sub-domain with~\eqref{eq:pmin_pa}.
If $p_a \leq \overline p$, then the matrices are well-conditioned according to~\eqref{eq:condition} and we define the localized approximating functions for the current sub-domain.
If $p_a > \overline p$, we partition~$\domain_a$ further into~$2^n$ uniform sub-domains.
We continue with the described procedure until we have constructed localized approximating functions that cover the whole domain.

\begin{algorithm}
\begin{algorithmic}[1]
\Require $f$, $\domain$, $\widetilde k$, $\error$, $\underline p$, $\overline{p}$
\State Init: $\mathcal{A}=\{0\}$, $\domain_0 = \domain$
\For{$a \in \mathcal{A}$} \Comment Dynamic for-loop
	\State Query oracle for $\rkhso$
 \Comment{Asm.~\ref{asm:oracle_RKHS_adaptive}}
\State Determine sufficiently large $p_a$ \Comment{ Prop.~\ref{pr:max_error},~\eqref{eq:pmin_pa}}
\If {$p_a \leq \overline p$} \Comment{\eqref{eq:condition}}
\State Get localized approximating functions
 \Comment{Alg.~\ref{alg:alki_function}}
\Else
    \State Define $2^n$ new sub-domains $\cup_{a^\prime \in \mathcal{A^\prime}} \domain_{a^\prime} = \domain_a$
    \State $\mathcal{A} \gets \mathcal{A} \cup \mathcal{A}^\prime \setminus a$ 
    \Comment{Update set partitions}
\EndIf
\EndFor 
\end{algorithmic}
\caption{Adaptive and Localized Kernel Interpolation}
\label{alg:alki_offline}
\end{algorithm}

The resulting localized approximating functions are used for the online evaluation, which is described in Algorithm~\ref{alg:alki_online}.
First, we define the function $h{:}\; \domain \rightarrow \mathbb{R}$ as the piecewise-defined approximating function on the domain~$\domain$, 
\ie
for any $a \in \mathcal{A}$, $c \in \mathcal{C}_{a,p_{a}}$, and $x \in \domain_c$, we define
\begin{align}\label{eq:h_after}
    \widetilde{h}(x)\coloneqq \widetilde h_{X_c}(x), \quad h(x) \coloneqq \widetilde h_{X_c}(x) + \mu_a. 
\end{align}

%
%
%
For the online evaluation, we determine the relevant local cube $\domain_c$ such that $x \in \domain_c$ by performing a tree search, and then evaluate the localized approximating function, similar to the online evaluation in~\cite{Lederer2020Realtime}.

\begin{algorithm}
\begin{algorithmic}[1]
\Require $\widetilde h$, $x$
        \State Tree search $c \in \cubes_{a,p_a}, a \in \mathcal{A}$ such that $x \in \domain_c \subseteq \domain_a \subseteq \domain$
        \State \Return  $h(x)$
\end{algorithmic}
\caption{Online Evaluation}
\label{alg:alki_online}
\end{algorithm}

\subsection{Sample complexity and uniform approximation error}
\label{sec:alki_theory}
To derive an upper bound on the maximum number of samples, we use the following assumption on the ground truth. 
\begin{assumption}\label{asm:oracle_RKHS_adaptive_bound}
There exists a function $\alpha \in \mathcal{K}_\infty$ such that for any $a\in\mathcal{A}$, the oracle in Assumption~\ref{asm:oracle_RKHS_adaptive} satisfies
\begin{align}\label{eq:RKHS_adaptive_bound}
     \alpha\left(\rkhso\right) \leq \ell_a.
\end{align}
\end{assumption}
\noindent
Shifted function values in combination with a continuous ground truth (see Lemma~\ref{le:cont}) ensure $\lim_{\ell_a \rightarrow 0} \widetilde f_a(x) = 0$ for all $x \in \domain_a$, which implies $\lim_{\ell_a \rightarrow 0} \|\widetilde f_a\|_\kla = 0$, and hence provides an intuition for Assumption~\ref{asm:oracle_RKHS_adaptive_bound}.
%
%
In fact, in the proof of Theorem~\ref{th:alkiax}, 
we show that the RKHS norm of the approximating function is upper-bounded by $\alpha^{-1}(\ell_a)$ with some later specified $\alpha \in \mathcal K_\infty$.

The following theorem provides an upper bound on the sample complexity~\ref{property:error_samples} and ensures that the localized approximating functions satisfy error bound~$\error$~\ref{property:error_satisfaction}.

\begin{theorem}\label{th:alki}
Let
Assumptions~\ref{asm:kernel}--\ref{asm:oracle_RKHS_adaptive_bound} hold.
Then, Algorithm~\ref{alg:alki_offline} terminates after querying at most 
$\mathrm{card}\left(X\right) \leq
\mathcal{O}\left(\left(\nicefrac{1}{\alpha(\error)}\right)^n\right)$ samples.
Furthermore, 
the resulting approximating function (Alg.\ \ref{alg:alki_online}) satisfies Inequality~\eqref{eq:uniform_error}, \ie the approximation error is uniformly bounded by $\error > 0$.
\end{theorem}
\begin{proof}
The proof is divided into three parts.
In Part~I, we show that Algorithm~\ref{alg:alki_offline} terminates once a threshold length scale is reached.
In Part~II, we derive the worst-case sample complexity in order to reach this threshold length scale.
In Part~III, we demonstrate that the resulting localized approximating functions satisfy error bound~$\error$.\\
\textit{Part I:} 
Consider any $a \in \mathcal{A}$. 
First, we show that Algorithm~\ref{alg:alki_offline} terminates whenever
\begin{align}\label{eq:ella_bound} 
    \ell_a \leq \alpha(\error).
\end{align}
For any $c \in \cubes_{a,p_a}$, it holds that
\begin{align} \label{eq:alkia_termination}
\rkhso
\overset{\eqref{eq:RKHS_adaptive_bound}}{\leq} \alpha^{-1}(\ell_a) \overset{\eqref{eq:ella_bound}}{\leq} \alpha^{-1}(\alpha(\error)) = \error,
\end{align}
which yields $p_a^*=0$ in Proposition~\ref{pr:max_error}.
Therefore,~\eqref{eq:pmin_pa} returns $p_a = \underline{p}\leq\overline p$, \ie  Algorithm~\ref{alg:alki_offline} terminates if \eqref{eq:ella_bound} is satisfied.
Thus, since the length scale $\ell_a$ is halved with each partitioning, 
\begin{align}\label{eq:ella_upper_bound}
    \ell_a > \frac{\alpha(\error)}{2}
\end{align}
always holds for any partition $a \in \mathcal{A}$.\\
\textit{Part II:}
The number of sub-domains satisfies 
\begin{align}\label{eq:cardA}
    {\mathrm{card}(\mathcal{A})} \leq \left(\min_{a \in \mathcal{A}} \ell_a\right)^{-n} \overset{\eqref{eq:ella_upper_bound}}{<} \left(\frac{2}{\alpha(\error)}\right)^n.
\end{align}
Moreover, each sub-domain $\domain_a$ has at most
\begin{align}\label{eq:Na}
    \mathrm{card}(X_a) \overset{\eqref{eq:condition}}{\leq} \left(1+2^{\overline p}\right)^n
\end{align}
samples.
Thus, Algorithm~\ref{alg:alki_offline} terminates after at most
\begin{align*}
    \mathrm{card}(X) &= \sum_{a\in\mathcal{A}} \mathrm{card}(X_a) \overset{\eqref{eq:cardA},\eqref{eq:Na}}{<} \left(\frac{2}{\alpha(\error)}\right)^n\left(1+2^{\overline{p}}\right)^n \\ &= 
\mathcal{O}\left(\left(\frac{1}{\alpha(\error)}\right)^n\right)
\end{align*}
samples. \\
\textit{Part~III:} 
Once Algorithm~\ref{alg:alki_offline} has terminated, each sub-domain~$\domain_a$ satisfies $p_a\geq p_a^*$.
Hence, Proposition~\ref{pr:max_error} ensures the satisfaction of error bound~\eqref{eq:uniform_error_cube}.
For any $c \in \mathcal{C}_{a,p_a}$, any $a \in \mathcal{A}$, and any $x \in \domain_c$, we have that $
\lvert f(x) - h(x) \rvert \overset{\eqref{eq:shifted_f},\eqref{eq:h_after}}{=} \lvert \widetilde f_a (x) - \widetilde h_{X_c}(x)\rvert \overset{\eqref{eq:uniform_error_cube}}{\leq} \error$, 
\ie error bound~\eqref{eq:uniform_error} holds.
\end{proof}

The approach proposed in this section leads to a one-shot algorithm with explicit worst-case sample complexity and a guaranteed error bound.
However, assuming a known RKHS norm of the unknown ground truth (see  Assumptions~\ref{asm:oracle_RKHS_adaptive} and~\ref{asm:oracle_RKHS_adaptive_bound}) can be restrictive in practice.
In the next section, we alleviate this issue by introducing a heuristic.

%% file: Sections/alkiax.tex

\section{\alkiax}
\label{sec:alkiax}
In this section, we extend the algorithm proposed in Section~\ref{sec:alki} to unknown RKHS norms by introducing a heuristic RKHS norm extrapolation.
This extension yields \alkiax, the \textbf{A}daptive and \textbf{L}ocalized \textbf{K}ernel \textbf{I}nterpolation \textbf{A}lgorithm with e\textbf{X}trapolated RKHS norm.
In Section~\ref{sec:alkiax_main_idea}, we introduce the main idea before explaining the details of the RKHS norm extrapolation in Section~\ref{sec:alkiax_detailed}.
We provide the theoretical analysis in Section~\ref{sec:alkiax_theory}, including the guaranteed error bound and the worst-case sample complexity.
Then, we present simpler complexity bounds for a popular kernel choice in Section~\ref{sec:se}.

\subsection{Main idea} \label{sec:alkiax_main_idea}
\alkiax\ has one major difference compared to the algorithm presented in Section~\ref{sec:alki}:
instead of assuming access to an oracle that returns an upper bound on the RKHS norm of the unknown ground truth (see Assumption~\ref{asm:oracle_RKHS_adaptive}), we introduce a novel \emph{RKHS norm extrapolation}~\ref{tool:RKHS}.
Specifically, we follow an extrapolation from the RKHS norm of the approximating function for each sub-domain $\domain_a$. 
\begin{figure}
    \centering
            \input{figures_final/RKHS.tex}
     \caption{Illustration of the RKHS norm extrapolation. 
     The magenta asterisks show the RKHS norm of the approximating function~$\|\widetilde h_{X_{a, p}}\|_\kla$~\eqref{eq:RKHS_h_a}.
     The solid black line is the extrapolating function~$\gamma_{a}$~\eqref{eq:extrapolation}, while the dashed blue line depicts its limit value~$\rkhsx$~\eqref{eq:gamma_bar}.
     }
    \label{fig:RKHS_samples}
\end{figure}
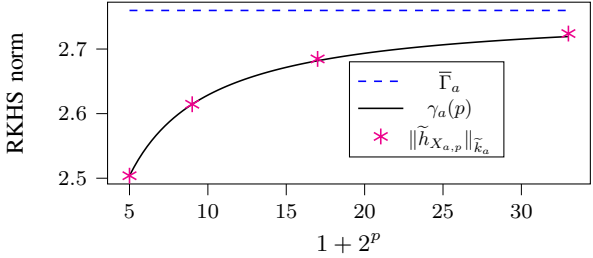
Figure~\ref{fig:RKHS_samples} shows that
the trend of the RKHS norm of the approximating function empirically resembles exponential behavior.\footnote{%
The results in Figure~\ref{fig:RKHS_samples} are generated using the Matérn kernel~\cite{Rasmussen2006Gaussian} with $\nu=\nicefrac{3}{2}$ and the function $f{:}\; [0,1]^2 \rightarrow \mathbb{R}$ with $f(x) = \sin(2\pi x_1) + \cos(2 \pi x_2)$, where $x=[x_1,x_2]^\top\in\mathbb{R}^2$ and $\domain_a = [0, \nicefrac{1}{3}]^2$.
The remaining hyperparameters are $\underline{p} = 2$, $\overline p = 5$ ($\overline{\kappa} = 1.14 \cdot 10^8$), and $\error = 10^{-6}$.}
Hence, we fit an \emph{exponential} function to the first two values of the RKHS norm of the approximating function and assume that its limit provides an upper bound on the RKHS norm of the unknown ground truth for each sub-domain.

\subsection{RKHS norm extrapolation} \label{sec:alkiax_detailed}
In the following, we formalize the exponential extrapolating function. 
To this end, we first introduce the RKHS norm of the shifted approximating function of sub-domain $\domain_a$ based on $\mathrm{card}(X_{a,\hat{p}})=(1+2^{\hat{p}})^n$ samples as
\begin{align}\label{eq:RKHS_h_a}
\|\widetilde{h}_{X_{a,\hat p}}\|_\kla = \sqrt{
\widetilde f_{X_{a,\hat p}}^\top K_{a, X_{a,\hat p}}^{-1} \widetilde f_{X_{a,\hat p}}
}
\quad \forall \hat p \in \{\underline p, \ldots,  \overline p\}.
\end{align}
Note that we only use $\|\widetilde{h}_{X_{a,\hat p}}\|_\kla$ for $\hat p \in \{\underline p, \underline p + 1\}$ to construct the extrapolating function.
To guarantee well-conditioned computations, we require that
\begin{align}\label{eq:condition_extended}
\mathrm{cond}(K_{a,X_{a,\hat p}})
\leq
\overline{\kappa} \quad \forall\hat p \in \{\underline p, \underline p+1\}, \; \forall  a \in \mathcal{A}
\end{align}
holds in addition to~\eqref{eq:condition}.
Note that~\eqref{eq:condition_extended} is enforced through the choice of~$\underline p$, and that~$\underline p$ and~$\overline \kappa$ have one-to-one correspondence and are uniform bounds.\footnote{%
Akin to~\eqref{eq:condition}, for a fixed $\underline p>0$, the condition number $\mathrm{cond}(K_{a,X_{a,\underline p+1}})$ is identical for each sub-domain $\domain_a$.
Moreover,~\eqref{eq:condition_extended} can be enforced a priori by checking whether the covariance matrix of the unscaled kernel $k$ (\ie with~$\ell=1$) with $\mathrm{card}(X)=(1+2^{\underline p +1})^n$ equidistant samples on domain $\domain = [0,1]^n$ is well-conditioned, \ie whether $\mathrm{cond}(K_{X}) \leq \overline\kappa$.
}

We upper-bound $\| \widetilde{h}_{X_{a,\hat p}} \|_\kla$,  $\hat p \in \{\underline{p}, \underline{p}+1\}$ with
\begin{align}\label{eq:gamma_hat}
    \hat{\gamma}_{a, \hat p} = \| \widetilde{h}_{X_{a,\hat p}}  \|_\kla + \frac{\error}{2^{\lambda+1}}, \quad \lambda \coloneqq 2 + 2^{-\underline p} .
\end{align}
The form of this upper bound becomes clearer when developing the theoretical guarantees in Section~\ref{sec:alkiax_theory}.
The extrapolating function $\gamma_{a}{:}\; \mathbb{N}_{\geq \underline{p}} \rightarrow \mathbb{R}_{\geq 0}$
is given by
\begin{align} \label{eq:extrapolation}
    \gamma_{a}(p) \coloneqq \rkhsx \exp\left(-\frac{\tau_{a}}{1+2^p}\right) \quad \forall p \in \{\underline{p}, \ldots, \overline{p}\},
\end{align}
where $\rkhsx$ is the limit of the extrapolating function and $\tau_{a}$ determines the rate at which the extrapolating function approaches its limit value.
\begin{lemma} \label{le:constants}
For any $a \in \mathcal{A}$ and any $\hat p \in \{\underline{p},\underline{p}+1\}$,
it holds that
\begin{align} \label{eq:interpolating}
    \gamma_{a}(\hat p) = \hat{\gamma}_{a,\hat p}
\end{align}
with
\begin{align} \label{eq:gamma_bar}
    \tau_{a} = \ln \left(\frac{\hat\gamma_{a,\underline{p}+1}}{\hat\gamma_{a, \underline{p}}}\right)\lambda
    \left(1+2^{\underline p}\right), \quad
    \rkhsx =\hat{\gamma}_{a, \underline p} \left( \frac{\hat{\gamma}_{a, \underline p+1}}{\hat{\gamma}_{a, \underline p}}\right)^{\lambda}
\end{align}
and $\rkhsx \in (0, \infty)$.
\end{lemma}
The proof can be found in Appendix~\ref{sec:appendix_lemma}.
Figure~\ref{fig:RKHS_samples} depicts the RKHS norm extrapolation heuristic applied to a toy example.
The resulting extrapolation approximates the trend of the RKHS norm of the approximating function well, although it has been constructed with the first two values $\|\widetilde h_{X_{a, p}}\|_\kla$, $p \in \{\underline p, \underline p+1\}$ only.
We use the limit of the extrapolating function $\rkhsx$ as an upper bound on the RKHS norm, which we formalize in Assumption~\ref{asm:extrapolation}.
\begin{assumption}\label{asm:extrapolation}
For any $a\in\mathcal{A}$, $\widetilde f_a$ is a member of the RHKS of kernel $\kla$. 
Moreover, the limit of the extrapolating function $\rkhso \in (0, \infty)$ (see~\eqref{eq:extrapolation}, \eqref{eq:gamma_bar}) is an upper bound on the RHKS norm, \ie $\rkhso \geq \|\widetilde f_a\|_\kla$.
\end{assumption}

\begin{remark}(RKHS norm approximation) \label{re:Gamma}
In general, it is not possible to find an upper bound on the RKHS norm of the ground truth  $\|f\|_k$ (or any shifted or restricted version of it) from finite samples, since it is an underlying characteristic of an unknown function.
Thus, there exists no perfect solution to tackle the challenging problem of working with the RKHS norm of an unknown ground truth.
A basic approach, followed by many related works (see~\eg\cite{\RKHSnormcite})
and also in Section~\ref{sec:alki} of this paper, is to assume that an oracle is available that returns a not overly conservative upper bound on the RKHS norm.
In order to not rely on an oracle, we introduce the RKHS norm extrapolation~\eqref{eq:extrapolation} with its corresponding limit~\eqref{eq:gamma_bar}.
The considered extrapolation is only a heuristic and Assumption~\ref{asm:extrapolation} may fail on some occasions.

Some early works in the field of information theory and signal detection compute the RKHS norm of the ground truth function.
However, the authors assume that an explicit expression of the ground truth is known
\cite{Kailath1971RKHS, Kailath1972RKHS}, which is not the case in the setting of this paper.
Moreover, in~\cite{Pradhan2022Submodular}, an upper bound on the RKHS norm of the ground truth based on the RKHS norm of the approximating function is derived.
However, the authors assume that the ground truth is a member of a pre-RKHS spanned by the chosen kernel (see \eg\cite[Section~2.3]{Kanagawa2018Gaussian}), which is more restrictive.

In the following, we briefly mention other approximations of
$\|f\|_k$ using $\|h_X\|_k$.
As noted in~\cite[Remark~2]{Maddalena2021Deterministic}, the RKHS norm of the approximating function $\|h_X\|_k$ is an underestimation of the RKHS norm of the ground truth $\|f\|_k$.
In~\cite[Appendix~A]{Scharnhorst2022Robust}, 
the authors use the intuition that $\|h_X\|_k$ may converge to $\|f\|_k$ for an increasing sample set $X$.
Based hereon,~\cite[Appendix~A]{Hashimoto2020Nonlinear} estimates an upper bound by \emph{manually} checking the convergence of $\|h_X\|_k$.
In contrast, we provide a simple formula to \emph{automatically} extrapolate an estimate of the RKHS norm,  which improves the ideas presented in~%
\cite{Scharnhorst2022Robust, Hashimoto2020Nonlinear} and empirically tends to closely match the underlying ground truth as seen in Figure~\ref{fig:RKHS_samples}.
\end{remark}

The offline approximation using \alkiax\ is summarized in Algorithm~\ref{alg:alkiax_offline}.
Compared to Algorithm~\ref{alg:alki_offline} in Section~\ref{sec:alki}, instead of receiving an upper bound on the RKHS norm of the ground truth from an oracle (see Assumption~\ref{asm:oracle_RKHS_adaptive}), we use the exponential extrapolation to compute~$\rkhsx$ (see  Assumption~\ref{asm:extrapolation}).
As a result, the exact number of required samples is determined during the approximation procedure, as~$\rkhsx$ is not known a priori.

\begin{algorithm}
\begin{algorithmic}[1]
\Require $f, \domain, \widetilde k, \error, \underline{p}, \overline{p}$ \Comment{\eqref{eq:condition_extended}}
\State Init: $\mathcal{A}=\{0\}$, $\domain_0 = \domain$
\For{$a \in \mathcal{A}$} \Comment{Dynamic for-loop}
\State Get samples $\widetilde f_{X_{a,\hat p}}, \; \hat p \in \{\underline p, \underline p+1\}$ \Comment{\eqref{eq:f},~\eqref{eq:shift}}
\State Compute $\rkhsx$ \Comment{\eqref{eq:RKHS_h_a},~\eqref{eq:gamma_hat},~\eqref{eq:gamma_bar}}
\State Determine sufficiently large $p_a$ \Comment{ Prop.~\ref{pr:max_error},~\eqref{eq:pmin_pa}}
  \If{$p_a \leq \overline p$} 
  \Comment{\eqref{eq:condition}}
  \State Get localized approximating function
  \Comment{Alg.~\ref{alg:alki_function}}
  \Else 
      \State Define $2^n$ new sub-domains $\cup_{a^\prime \in \mathcal{A^\prime}} \domain_{a^\prime} = \domain_a$
    \State $\mathcal{A} \gets \mathcal{A} \cup \mathcal{A}^\prime \setminus a$ 
    \Comment{Update set partitions}
  \EndIf
\EndFor
\end{algorithmic}
\caption{\alkiax}
\label{alg:alkiax_offline}
\end{algorithm}

\subsection{Sample complexity and uniform approximation error}\label{sec:alkiax_theory}

The following theorem provides a bound on the sample complexity and ensures the satisfaction of error bound~$\error$.
\begin{theorem}\label{th:alkiax}
Let Assumptions~\ref{asm:kernel},~\ref{asm:max_power}, and~\ref{asm:extrapolation} hold.
Then, Algorithm~\ref{alg:alkiax_offline} terminates after at most 
\begin{align}
\mathrm{card}(X)\leq \mathcal{O}\left(\left(
\frac{\sqrt{n}}{\kinv\left(C\left(\frac{\error}{ \|f\|_k
} \right)^2
\right)}
\right)^n
\right)
\end{align}
samples, 
with some $C > 0$ independent of $\error$, $f$.
Furthermore, the resulting approximating function 
(Alg.~\ref{alg:alki_online}) satisfies Inequality~\eqref{eq:uniform_error}, \ie the approximation error is uniformly bounded by $\error > 0$.
\end{theorem}
\begin{proof}
The proof is structured analogously to the proof of Theorem~\ref{th:alki}.
In Part~I, we show that Algorithm~\ref{alg:alkiax_offline} terminates once a threshold length scale is reached. 
In Part~II, we derive the worst-case sample complexity in order to reach that threshold length scale.
In Part~III, we demonstrate that the resulting localized approximating functions satisfy error bound~$\error$.\\
\textit{Part I:} Consider any $a \in \mathcal{A}$.
First, we prove that $\lvert \widetilde f_a(x) \rvert \leq \|f\|_k \sqrt{2\knorm(\sqrt n\ell_a)}$ holds for all $x \in \domain_a$.
Consider an arbitrary $x \in \domain_a$, where, w.l.o.g., $\widetilde f_a(x) \geq 0$.
Pick any other $x^\prime \in \domain_a$ with $\widetilde f_a(x^\prime) \leq 0$, which exists since ${\sum_{x \in X_{a,\underline{p}}} \widetilde f_a(x)}=0$ (see~\eqref{eq:shift}).
It holds that 
 \begin{align}\label{eq:lipschitz_bound}
     \lvert \widetilde f_a(x) \rvert 
     &\leq \lvert \widetilde f_a(x) - \widetilde f_a(x^\prime) \rvert
     \overset{\eqref{eq:shifted_f}}{=} \lvert f(x) - f(x^\prime) \rvert  \\
     \overset{\mathrm{Lem.\ref{le:cont}}}&{\leq} \|f\|_k \sqrt{2\knorm(\|x-x^\prime\|_2)} 
     \overset{\eqref{eq:scaled_kernel}}{\leq}
     \|f\|_k \sqrt{2\knorm(\sqrt{n}\ell_a)} \nonumber.  
\end{align}
Next, we derive an upper bound on $\|\widetilde h_{X_{a,\hat p}}\|_\kla$.
For all $\hat p \in \{\underline{p}, \underline{p}+1\}$, it holds that
\begin{align} \label{eq:h_RKHS_bound}
    \|\widetilde h_{X_{a,\hat p}}\|_\kla \overset{\eqref{eq:h_RKHS}}&{=}
    \sqrt{\widetilde f_{X_{a,\hat p}}^\top K_{a,X_{a,\hat p}}^{-1}\widetilde f_{X_{a,\hat p}}} \nonumber \\
    &\leq
    \sqrt{\|\widetilde f_{X_{a,\hat p}}\|_1 \|\widetilde f_{X_{a,\hat p}}\|_\infty   \|K_{a,X_{a,\hat p}}^{-1}\|_\infty   
    } \\
    \overset{\eqref{eq:lipschitz_bound}}&{\leq}
     \|f\|_k \sqrt{2\knorm(\sqrt n\ell_a) \left(1+2^{\underline{p}+1}\right)^n \|K_{a,X_{a,\hat p}}^{-1}\|_\infty
    } \nonumber \\
    \overset{\eqref{eq:condition_extended}}&{\leq}
     \|f\|_k \sqrt{2\knorm(\sqrt n \ell_a) \left(1+2^{\underline{p}+1}\right)^n
    \overline{\kappa}
    }. \nonumber
\end{align}
The first inequality follows from  Hölder's inequality, whereas the second inequality holds since $\mathrm{card}(X_{a,\hat p}) \leq \left(1+2^{\underline p +1}\right)^n$ for all $\hat p \in\{\underline p, \underline p+1\}$.
Moreover, the last inequality is satisfied since $\|K_{a,X_{a,\hat p}}^{-1}\|_\infty = \frac{\mathrm{cond}(K_{a,X_{a,\hat p}})}{\|K_{a,X_{a,\hat p}}\|_\infty} \overset{\eqref{eq:condition_extended}}{\leq}  \overline \kappa$ with $\|K_{a,X_{a,\hat p}}\|_\infty \geq \max_{i,j} K_{a,X_{a,p}}(i,j)=\widetilde{k}(0)=1$. 
In the following, we show that Algorithm~\ref{alg:alkiax_offline} terminates whenever
\begin{align}\label{eq:alkiax_ell_bound}
\ell_a \leq \frac{ \kinv \left(
\frac{\error^2}{2 \overline \kappa \left( 1+2^{\underline p+1}\right)^n
2^{2\lambda +2} \|f\|_k^2
}
\right)}
{\sqrt{n}}.
\end{align}
For any $\hat p \in \{\underline p, \underline p+1\}$, \eqref{eq:h_RKHS_bound} and \eqref{eq:alkiax_ell_bound} imply
\begin{align} \label{eq:h_norm_bound}
    \| \widetilde h_{X_{a,\hat p}} \|_\kla \leq \frac{\error}{2^{\lambda + 1}}.
\end{align}
For any $c \in \cubes_{a,p_a}$, \eqref{eq:h_norm_bound} implies
\begin{align}\label{eq:alkiax_error}
    &
    \rkhsx
    \overset{\eqref{eq:gamma_bar}}{=} \hat{\gamma}_{a,\underline p} \left( \frac{\hat{\gamma}_{a,\underline p+1}}{\hat{\gamma}_{a,\underline p}}\right)^{\lambda} 
    \nonumber  \\
    \overset{\eqref{eq:gamma_hat}}&{=} \left(\| \widetilde{h}_{X_{a,\underline{p}}}  \|_\kla + \frac{\error}{2^{\lambda+1}}
 \right) \left(\frac{\| \widetilde{h}_{X_{a,\underline{p}+1}}  \|_\kla + \frac{\error}{2^{\lambda +1}}
}{\| \widetilde{h}_{X_{a,\underline{p}}}  \|_\kla + \frac{\error}{2^{\lambda +1}}
}\right)^\lambda 
 \\  
    &\leq \left(\| \widetilde{h}_{X_{a,\underline{p}}}  \|_\kla + \frac{\error}{2^{\lambda +1}}
 \right) \left(\frac{\error + \| \widetilde{h}_{X_{a,\underline{p}+1}}  \|_\kla 2^{\lambda + 1}}{ \error}\right)^\lambda 
\nonumber \\
    \overset{\eqref{eq:h_norm_bound}}&{\leq} \left(\frac{\error}{2^{\lambda + 1}} + \frac{\error}{2^{\lambda +1}} \right) \left( \frac{\error + \error }{\error}\right)^\lambda 
    = \left(\frac{2\error}{2^{\lambda +1}}\right)2^\lambda 
    = \error. \nonumber
\end{align}
Thus, analogous to Part~I of the proof of Theorem~\ref{th:alki}, this yields $p_a^*=0$ in Proposition~\ref{pr:max_error} and $p_a = \underline{p}\leq\overline p$ in~\eqref{eq:pmin_pa}, wherefore Algorithm~\ref{alg:alkiax_offline} terminates if \eqref{eq:alkiax_ell_bound} is satisfied.
Thus,
\begin{align}\label{eq:alkiax_ell_termination}
\ell_a > \frac{ \kinv \left(
\frac{\error^2}{2 \overline \kappa \left( 1+2^{\underline p+1}\right)^n
2^{2\lambda +2} \|f\|_k^2
}
\right)}
{2\sqrt{n}}
\end{align}
always holds.
Hence, we have that 
\begin{align} \label{eq:ell_a_min}
     \ell_a \leq \mathcal{O}\left(
    \frac{\kinv\left(C\left(\frac{\error}{\|f\|_k}\right)^2\right)}{\sqrt{n}}
    \right)\quad \forall a \in \mathcal{A},
\end{align}
where $C>0$ is a constant that is independent of $\error$, $f$.
\\
\textit{Part II:}
From Part~II of the proof of Theorem~\ref{th:alki}, we know that $\mathrm{card}(X) \leq (\min_{a \in \mathcal{A}} \ell_a)^{-n} \left(1+2^{\overline p}\right)^n$.
With~\eqref{eq:ell_a_min}, it follows that
\begin{align}  \label{eq:alkiax_sample_complexity}
        \mathrm{card}(X) &\leq\mathcal{O}
        \left(
        \left(\frac{\sqrt{n}(1+2^{\overline p})}{
\kinv\left(
C
\left(\frac{\error}{ \|f\|_k
} \right)^2
\right)}        \right)^n 
\right)
\\
 &=  \mathcal{O}\left(\left(
\frac{\sqrt{n}}{
\kinv\left(
C
\left(\frac{\error}{ \|f\|_k
} \right)^2
\right)}
\right)^n
\right). \nonumber
\end{align}
\\
\textit{Part III:} 
The proof is analogous to Part~III of the proof of Theorem~\ref{th:alki}.
We determine $\rkhso$ with Assumption~\ref{asm:extrapolation} and receive $p_a$ from Proposition~\ref{pr:max_error} and~\eqref{eq:pmin_pa}.
The satisfaction of Inequality~\eqref{eq:uniform_error} directly follows.
\end{proof}

Since $\kinv (0) = 0$, we have that~$\nicefrac{\error}{\|f\|_k} \rightarrow 0$ implies~$\mathrm{card}(X) \rightarrow \infty$. 
This result is intuitive, as more complex functions (corresponding to a larger RKHS norm~$\|f\|_k$) or a smaller allowed approximation error~$\error$ require more samples.
Furthermore, the theoretical results obtained in this section show that, in the worst case, the number of samples to terminate \alkiax\ may scale exponentially with the input dimension, which is expected and widely known as the curse of dimensionality. 

\subsection{Complexity bounds for the SE-kernel}\label{sec:se}
In this section, we derive more intuitive bounds on the online and offline complexity for the special case of the squared-exponential (SE) kernel, which is one the the most common kernels.
The following lemma provides a bound on the smallest length scale that can occur when executing~Algorithm~\ref{alg:alkiax_offline} (see   Part~II of the proof of Theorem~\ref{th:alkiax}). 

\begin{lemma}\label{le:SE}
Suppose that the assumptions in Theorem~\ref{th:alkiax} hold and that~$k$ is the SE-kernel.
Then, for any $a\in\domain_a$, it holds that
\begin{align}\label{eq:ell_limit}
    \ell_a \leq \mathcal{O} \left(
    \frac{ \error}{\sqrt n\|f\|_k}    
    \right)
    \quad \mathrm{as}\; \frac{\error}{\|f\|_k}\rightarrow 0.
\end{align}
\end{lemma}
The proof can be found in Appendix~\ref{sec:appendix_lemma_SE_limit}.
Moreover, the following corollary provides an upper bound on the number of samples sufficient to terminate Algorithm~\ref{alg:alkiax_offline}.

\begin{corollary}\label{co:SE}
Suppose that the conditions in Lemma~\ref{le:SE} hold.
Then, it holds that
\begin{align}\label{eq:se_limit}
    \mathrm{card}(X) \leq \mathcal{O} \left(\left(\frac{ \sqrt n\|f\|_k}{\error}\right)^n\right) \quad \mathrm{as}\; \frac{\error}{\|f\|_k}\rightarrow 0.
\end{align}
\end{corollary}
\begin{proof}
From Part~II of Theorem~\ref{th:alki}, we know that $\mathrm{card}(X) \leq (\min_{a\in\mathcal A}\ell_a)^{-n} \left(1+2^{\overline p}\right)^n = \mathcal{O}((\min_{a\in\mathcal A}\ell_a)^{-n})$.
For $\nicefrac{\error}{\|f\|_k}\rightarrow 0$, Equation~\eqref{eq:ell_limit} derived in Lemma~\ref{le:SE} 
provides a bound on $\ell_a $ for all  $\in \mathcal A$,
from which~\eqref{eq:se_limit} directly follows. 
\end{proof} 



\subsubsection*{Offline computational complexity}
Inequality~\eqref{eq:se_limit} shows that, for a given $\nicefrac{\error}{\|f\|_k}$, the number of maximum samples to terminate Algorithm~\ref{alg:alkiax_offline} scales exponentially with the input dimension $n$.
Since we store all samples $X$ to generate the approximating functions (see  Algorithm~\ref{alg:alki_function}), the memory requirements also scale exponentially with $n$.

\subsubsection*{Online computational complexity}
The online evaluation with Algorithm~\ref{alg:alki_online} for an input $x\in \domain$ can be divided into three steps.
First, we execute a tree search to find $\domain_c$ such that $x \in \domain_c \subseteq \domain_a \subseteq \domain$.
Second, we construct the covariance vector~\eqref{eq:k_x}.
Third, we evaluate the approximating function~\eqref{eq:shifted_h_cube}.
Next, we elaborate on the complexity of each step.

Each sub-domain $\domain_a \subseteq \domain$ corresponds to a node of the resulting tree of the 
offline approximation. 
The domain $\domain$ corresponds to the root node and the sub-domains containing the local cubes with the localized approximating functions are located in the leaf nodes.
The complexity of the tree search scales linearly with the depth of the tree.
Since the length scale is halved from one depth level to the next one, the maximum depth $\overline d \in \mathbb{N}$ is given by
\begin{align} \label{eq:depth}
    \overline d = \mathcal{O}\left(\log_2\left(\left(\min_{a\in \mathcal{A}} \ell_a\right)^{-1}\right)\right) \overset{\eqref{eq:ell_limit}}{\leq} \mathcal{O}\left(\log_2\left(\frac{\|f\|_k \sqrt n}{\error}\right)\right).
\end{align}
Hence, the computational complexity of the tree search only scales \emph{logarithmically}.
This result is comparable to the logarithmic complexity of localized GP regression obtained in~\cite{Lederer2020Realtime}.

Since we use local approximations based on local cubes with $2^n$ samples, each covariance vector~\eqref{eq:k_x} is $2^n$-dimensional.
Thus, to construct the covariance vector, the kernel needs to be evaluated at $2^n$ points.

In~\eqref{eq:shifted_h_cube}, the vector-matrix multiplication~$f_{X_c}^\top K_{a,X_{c}}^{-1}$ can be computed offline, and hence evaluating the approximating function~\eqref{eq:shifted_h_cube} online reduces to computing $2^n \times 2^n$ vector-vector multiplications.

To conclude, although the offline sample complexity and the memory requirements scale exponentially with the input dimension~$n$ for a given $\nicefrac{\|f\|_k}{\error}$, the online evaluation is characterized by operations of order $2^n$ that are \emph{independent} of the desired accuracy~$\error$ and the ground truth~$f$, and the tree search that scales \emph{logarithmically} with $\nicefrac{\|f\|_k \sqrt n}{\error}$.

%% file: figures_final/RKHS.tex
\begin{tikzpicture}[yscale=1]
\definecolor{darkgray176}{RGB}{176,176,176}
\pgfplotsset{
every axis legend/.append style={
at={(0.5,0.675)},
anchor=north west,
},
}
\begin{axis}[
tick align=outside,
tick pos=left,
x grid style={darkgray176},
xlabel={\small \(\displaystyle 1+2^p\)},
xmin=3.6, xmax=34.4,
xtick style={color=black},
y grid style={darkgray176},
ylabel={\small RKHS norm},
height=4cm,
width=8cm,
ymin=2.49110171352698, ymax=2.77252778180555,
ytick style={color=black},
legend style={nodes={scale=0.75}},
tick label style={font=\footnotesize}
]
\addplot [dashed, semithick, blue]
table {%
5 2.75973568779288
5.28282828282828 2.75973568779288
5.56565656565657 2.75973568779288
5.84848484848485 2.75973568779288
6.13131313131313 2.75973568779288
6.41414141414141 2.75973568779288
6.6969696969697 2.75973568779288
6.97979797979798 2.75973568779288
7.26262626262626 2.75973568779288
7.54545454545454 2.75973568779288
7.82828282828283 2.75973568779288
8.11111111111111 2.75973568779288
8.39393939393939 2.75973568779288
8.67676767676768 2.75973568779288
8.95959595959596 2.75973568779288
9.24242424242424 2.75973568779288
9.52525252525253 2.75973568779288
9.80808080808081 2.75973568779288
10.0909090909091 2.75973568779288
10.3737373737374 2.75973568779288
10.6565656565657 2.75973568779288
10.9393939393939 2.75973568779288
11.2222222222222 2.75973568779288
11.5050505050505 2.75973568779288
11.7878787878788 2.75973568779288
12.0707070707071 2.75973568779288
12.3535353535354 2.75973568779288
12.6363636363636 2.75973568779288
12.9191919191919 2.75973568779288
13.2020202020202 2.75973568779288
13.4848484848485 2.75973568779288
13.7676767676768 2.75973568779288
14.0505050505051 2.75973568779288
14.3333333333333 2.75973568779288
14.6161616161616 2.75973568779288
14.8989898989899 2.75973568779288
15.1818181818182 2.75973568779288
15.4646464646465 2.75973568779288
15.7474747474747 2.75973568779288
16.030303030303 2.75973568779288
16.3131313131313 2.75973568779288
16.5959595959596 2.75973568779288
16.8787878787879 2.75973568779288
17.1616161616162 2.75973568779288
17.4444444444444 2.75973568779288
17.7272727272727 2.75973568779288
18.010101010101 2.75973568779288
18.2929292929293 2.75973568779288
18.5757575757576 2.75973568779288
18.8585858585859 2.75973568779288
19.1414141414141 2.75973568779288
19.4242424242424 2.75973568779288
19.7070707070707 2.75973568779288
19.989898989899 2.75973568779288
20.2727272727273 2.75973568779288
20.5555555555556 2.75973568779288
20.8383838383838 2.75973568779288
21.1212121212121 2.75973568779288
21.4040404040404 2.75973568779288
21.6868686868687 2.75973568779288
21.969696969697 2.75973568779288
22.2525252525253 2.75973568779288
22.5353535353535 2.75973568779288
22.8181818181818 2.75973568779288
23.1010101010101 2.75973568779288
23.3838383838384 2.75973568779288
23.6666666666667 2.75973568779288
23.9494949494949 2.75973568779288
24.2323232323232 2.75973568779288
24.5151515151515 2.75973568779288
24.7979797979798 2.75973568779288
25.0808080808081 2.75973568779288
25.3636363636364 2.75973568779288
25.6464646464646 2.75973568779288
25.9292929292929 2.75973568779288
26.2121212121212 2.75973568779288
26.4949494949495 2.75973568779288
26.7777777777778 2.75973568779288
27.0606060606061 2.75973568779288
27.3434343434343 2.75973568779288
27.6262626262626 2.75973568779288
27.9090909090909 2.75973568779288
28.1919191919192 2.75973568779288
28.4747474747475 2.75973568779288
28.7575757575758 2.75973568779288
29.040404040404 2.75973568779288
29.3232323232323 2.75973568779288
29.6060606060606 2.75973568779288
29.8888888888889 2.75973568779288
30.1717171717172 2.75973568779288
30.4545454545455 2.75973568779288
30.7373737373737 2.75973568779288
31.020202020202 2.75973568779288
31.3030303030303 2.75973568779288
31.5858585858586 2.75973568779288
31.8686868686869 2.75973568779288
32.1515151515151 2.75973568779288
32.4343434343434 2.75973568779288
32.7171717171717 2.75973568779288
33 2.75973568779288
};
\addplot [semithick, black]
table {%
5 2.5038939812464
5.28282828282828 2.5169695972137
5.56565656565657 2.52877450901235
5.84848484848485 2.53948525356251
6.13131313131313 2.54924713277596
6.41414141414141 2.55818082518138
6.6969696969697 2.56638738724858
6.97979797979798 2.57395208389194
7.26262626262626 2.58094735576854
7.54545454545454 2.58743514200267
7.82828282828283 2.59346871591901
8.11111111111111 2.59909414884093
8.39393939393939 2.60435148696904
8.67676767676768 2.60927570485634
8.95959595959596 2.61389748342421
9.24242424242424 2.6182438490575
9.52525252525253 2.62233870187473
9.80808080808081 2.62620325495916
10.0909090909091 2.62985640157722
10.3737373737374 2.63331502378921
10.6565656565657 2.63659425308012
10.9393939393939 2.63970769149183
11.2222222222222 2.64266760006703
11.5050505050505 2.64548506010581
11.7878787878788 2.64817011170285
12.0707070707071 2.65073187321332
12.3535353535354 2.65317864464145
12.6363636363636 2.6555179974201
12.9191919191919 2.65775685262622
13.2020202020202 2.65990154933322
13.4848484848485 2.66195790452121
13.7676767676768 2.66393126573692
14.0505050505051 2.6658265575066
14.3333333333333 2.66764832234968
14.6161616161616 2.66940075711189
14.8989898989899 2.6710877452294
15.1818181818182 2.67271288544579
15.4646464646465 2.67427951742855
15.7474747474747 2.67579074466882
16.030303030303 2.67724945499455
16.3131313131313 2.67865833898237
16.5959595959596 2.68001990651493
16.8787878787879 2.6813365016982
17.1616161616162 2.68261031632502
17.4444444444444 2.68384340204757
17.7272727272727 2.68503768140099
18.010101010101 2.68619495780254
18.2929292929293 2.68731692463586
18.5757575757576 2.68840517351632
18.8585858585859 2.68946120182251
19.1414141414141 2.69048641956854
19.4242424242424 2.69148215568375
19.7070707070707 2.69244966375837
19.989898989899 2.69339012730753
20.2727272727273 2.69430466459997
20.5555555555556 2.69519433309292
20.8383838383838 2.69606013351017
21.1212121212121 2.69690301359623
21.4040404040404 2.6977238715764
21.6868686868687 2.69852355934908
21.969696969697 2.69930288543434
22.2525252525253 2.7000626177001
22.5353535353535 2.7008034858853
22.8181818181818 2.70152618393742
23.1010101010101 2.70223137218016
23.3838383838384 2.70291967932542
23.6666666666667 2.7035917043426
23.9494949494949 2.70424801819674
24.2323232323232 2.7048891654662
24.5151515151515 2.70551566584954
24.7979797979798 2.70612801557021
25.0808080808081 2.70672668868721
25.3636363636364 2.70731213831891
25.6464646464646 2.7078847977866
25.9292929292929 2.70844508168401
26.2121212121212 2.70899338687818
26.4949494949495 2.70953009344679
26.7777777777778 2.71005556555674
27.0606060606061 2.71057015228806
27.3434343434343 2.71107418840711
27.6262626262626 2.71156799509284
27.9090909090909 2.71205188061914
28.1919191919192 2.71252614099658
28.4747474747475 2.71299106057615
28.7575757575758 2.71344691261778
29.040404040404 2.71389395982579
29.3232323232323 2.71433245485377
29.6060606060606 2.71476264078058
29.8888888888889 2.71518475155974
30.1717171717172 2.71559901244358
30.4545454545455 2.71600564038407
30.7373737373737 2.71640484441162
31.020202020202 2.71679682599336
31.3030303030303 2.71718177937216
31.5858585858586 2.71755989188751
31.8686868686869 2.71793134427956
32.1515151515151 2.71829631097712
32.4343434343434 2.71865496037075
32.7171717171717 2.71900745507186
33 2.71935395215844
};
\addplot [semithick, magenta, mark=asterisk, mark size=3, mark options={solid}, only marks]
table {%
5 2.50389380753964
9 2.61453449203542
17 2.6846997801278
33 2.72397880380792
};

\legend{
$\rkhsx$,
$\gamma_{a}(p)$,
$\|\widetilde h_{X_{a, p}}\|_\kla$
}

\end{axis}

\end{tikzpicture}

%% file: Sections/MPC.tex

\section{
Approximate MPC via \alkiax 
} \label{sec:alkiax_MPC}

In this section, we investigate the characteristics of the approximate MPC obtained using \alkiax.
As introduced in Section~\ref{sec:problem_setting}, the ground truth~$f$ corresponds to the MPC feedback law  $u=f(x)$ 
and \alkiax\ determines an explicit feedback law $u=h(x)$ that approximates the MPC feedback~$f$ on the domain~$\domain$.
\subsubsection*{
Infeasible states}
\alkiax\ requires a cubic domain~$\domain$ and the evaluation of the ground truth $f$ for any $x \in \domain$.
The set $\domain$ should represent the state constraints, and in the standard case of box constraints, they can be shifted and scaled to yield a cubic set $\domain$.
The evaluation of $f$ requires the solution of the underlying NLP, which is only defined on some feasible set~$\xfeasf\subseteq \domain$. 
We define function evaluations for infeasible states $x \in \domain \setminus \xfeasf$ by relaxing hard state and terminal constraints in the MPC formulation using slack variables and penalties~\cite{Kerrigan2000Soft}. 
Then, under suitable regularity conditions, the optimization problem returns the same solution as the original MPC for all feasible states $x \in \xfeasf$ 
(see~\cite{Rosenberg1984Exact}),
while also providing a well-defined solution for $x \in \domain \setminus \xfeasf$.
\subsubsection*{
Feasible domain of the approximate MPC}
Although we can obtain function evaluations for the whole domain~$\domain$ during the
offline approximation, there is no need for approximating the MPC on infeasible parts of
the domain $\domain$.
If each state within a sub-domain results in a non-zero slack variable, \alkiax\ classifies that sub-domain as infeasible and does not continue the approximation process on that sub-domain.
Hence, \alkiax\ defines its own feasible domain $\xfeash$ on which it computes an approximating function with the desired approximation accuracy~$\error$. 
We assume that $\xfeasf \subseteq \xfeash$ holds to ensure that the error bound $\error$ is enforced for all feasible states $x \in \xfeasf$.
\subsubsection*{Non-scalar ground truths}
In Sections~\ref{sec:alki} and~\ref{sec:alkiax}, we developed \alkiax\ with its corresponding guarantees for a scalar ground truth $f{:}\; \domain\rightarrow \mathbb{R}$.
However, as discussed in Section~\ref{sec:problem_setting}, the MPC feedback law $f{:}\; \domain\rightarrow \mathbb{R}^{n_u}$ is in general non-scalar, \ie $n_u > 1$.
We can approximate such non-scalar functions by treating each dimension separately.
Specifically, in Algorithm~\ref{alg:alkiax_offline}, we determine the sufficiently large integer~$p_a$ corresponding to the sufficient number of samples as the maximum of the sufficient number of samples required to approximate each control input dimension separately.
Moreover, the localized approximating functions for each dimension are constructed (see  Algorithm~\ref{alg:alki_function}) and evaluated (see  Algorithm~\ref{alg:alki_online}) independently.
Hence, \alkiax\ ensures $\|f(x)-h(x)\|_\infty\leq \error$ for all $x\in\xfeasf\subseteq\xfeash$.

\subsubsection*{Closed-loop guarantees}\label{sec:closed-loop}
We provide closed-loop guarantees on the approximate MPC by combining a robust MPC design with the guaranteed approximation error (see  Figure~\ref{fig:framework}).
To this end, consider a disturbance set $\mathcal{W}\subseteq \mathbb{R}^n$ that upper-bounds the effect of inexactly approximating the control input, \ie
$\mathcal{W}\supseteq\{g(x,\tilde{u})-g(x,u)|~x\in\domain,u\in\mathcal{U},\|u-\tilde{u}\|_\infty\leq\error\}$, see \eg constructions in~\cite{Hertneck2018Learning,Nubert2020Robot}.

\begin{corollary}\label{co:MPC}
Let~$f$ be according to the robust MPC design in~\cite[Problem~(30)]{Kohler2021Uncertain} with disturbance bound $\mathcal{W}$ and suppose that the conditions in~\cite[Theorem~2]{Kohler2021Uncertain} regarding $f,g,\domain,\mathcal{U},\mathcal{W}$ hold.\footnote{%
As discussed above,~$f$ is the robust MPC design with constraints relaxed using penalties to ensure that the function is defined for all $x\in\domain$. We assume that the conditions in~\cite{Rosenberg1984Exact} hold, ensuring that~$f(x)$ coincides with the minimizer of~\cite[Problem~(30)]{Kohler2021Uncertain} for all $x\in\xfeasf$.}
Suppose further that the conditions in Theorem~\ref{th:alkiax} hold and that~$h$ is constructed according to Algorithm~\ref{alg:alkiax_offline}, where we approximate each dimension independently.
Then, for any initial condition $x_0\in\xfeasf$, the closed-loop system $x_{t+1}=g(x_t,u_t)$, $u_t=h(x_t)$ ensures
\begin{enumerate}[(i)]
\item recursive feasibility $x_t\in \xfeasf\subseteq \xfeash, \; \forall t\in\mathbb{N}$,
\item constraint satisfaction $x_t\in\domain$, $u_t\in\mathcal{U}$, $\forall t\in\mathbb{N}$, 
\item and practical\footnote{%
     Practical asymptotic stability implies convergence to a neighborhood of the steady-state $x_\mathrm{s}$. In this case, the size of the neighborhood depends on the magnitude of the approximation error 
     $\max_{x\in\xfeasf}\| f(x)-h(x)\|_\infty \leq \error$.} asymptotic stability of $x=x_\mathrm{s}$.
\end{enumerate}
\end{corollary}
\begin{proof}
Theorem~\ref{th:alkiax} ensures that each dimension of the approximating function~$h$, which we receive independently from Algorithm~\ref{alg:alkiax_offline}, satisfies Inequality~\eqref{eq:uniform_error}, \ie the satisfaction of $\|f(x)-h(x)\|_\infty\leq\epsilon$ is ensured for all $x\in\xfeash$.
    Hence, the closed-loop system satisfies $x_{t+1}=g(x_t,h(x_t))=:g(x_t,f(x_t))+w_t$ with $w_t\in\mathcal{W}$ for all $x\in\xfeasf\subseteq\xfeash$. Thus, the closed-loop guarantees in~\cite[Theorem~2]{Kohler2021Uncertain} with respect to disturbances $w_t\in\mathcal{W}$ apply.
\end{proof}
    These theoretical guarantees are equally applicable if any other nonlinear MPC scheme is approximated
    that is designed to be robust with respect to disturbances $w_t\in\mathcal{W}$, see \eg \cite{\robustMPCcite} for corresponding designs.

\subsubsection*{Computational complexity}\label{sec:scalability}
In the complexity analysis in Section~\ref{sec:se}, we showed that the offline computational complexity of \alkiax\ scales exponentially with the number of states $n$. 
In contrast, the computational complexity of classical explicit MPC approaches scales with the number of states, the number of constraints, and the prediction horizon~\cite[Section~3]{Alessio2009Explicit}.
Hence, a major benefit of approximating MPC schemes using \alkiax\ is the fact that the computational complexity does \emph{not} directly increase with the prediction horizon or the number of constraints used in the MPC formulation.
Moreover, the online computational complexity of classical kernel-based methods scales cubicly in the number of samples \cite[Section~2.3]{Rasmussen2006Gaussian}.
In contrast, the online computational complexity of \alkiax\ only scales \emph{logarithmically} with $\nicefrac{\|f\|_k\sqrt n}{\error}$.
Thus, the resulting approximate MPC is fast-to-evaluate, and hence suitable to control nonlinear systems that require high sampling rates with relatively cheap hardware~\ref{property:fast}.

%% file: Sections/numerical_experiment.tex
\section{Numerical Experiments} \label{sec:numerical}
In this section, we approximate two nonlinear MPC schemes using \alkiax.
In Section~\ref{sec:cstr}, we apply \alkiax\ on an academic example and compare its performance with existing work.
In Section~\ref{sec:plasma}, we demonstrate the practicability of \alkiax\ by approximating 
an MPC corresponding to a real-world application for which fast sampling rates are required and computational capacity is limited. 
\alkiax\ and both MPCs are implemented in \python.
The MPC schemes are formulated using \casadi~\cite{Casadi2019} and the underlying NLPs are solved with IPOPT~\cite{IPOPT2006} using just-in-time compilation.
We conducted the offline approximations on a Linux server with \SI{32}{\giga\byte} RAM
parallelized on~8 cores.
The online evaluation was performed using an Ubuntu laptop with \SI{32}{\giga\byte} RAM and an Intel Core i7-12700H processor.
Moreover, we used the Matérn kernel with $\nu=\nicefrac{3}{2}$ and a length scale of~$0.8$ for both experiments.

\subsection{Continuous stirred tank reactor}\label{sec:cstr}
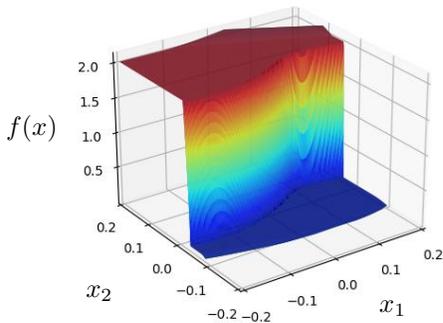
\begin{figure}
    \centering
    \input{figures_final/ground_truth.tex}
    \caption{Ground truth MPC feedback law~$f$ over the domain $\domain=[-0.2,0.2]^2$. 
}
    \label{fig:ground_truth}
\end{figure}

We consider a nonlinear continuous stirred tank reactor with
\begin{align}\label{eq:CSTR}
    g( x, u)&=\begin{pmatrix}
         x_1 + \delta \left(\frac{1- x_1}{\theta}- \hat k  x_1 e^{-\frac{M}{ x_1}}\right) \\
         x_2+\delta\left(\frac{x_f- x_2}{\theta}+\hat k x_1e^{-\frac{M}{ x_2}}-\hat \alpha  u( x_2-x_c)
        \right)
            \end{pmatrix},
\end{align}
where~$x_1$ is the temperature,~$x_2$ is the concentration,~$u$ is the coolant flow, and~$\delta=\SI{0.1}{\second}$ is the sampling time, see~\cite{Hertneck2018Learning} for details.
The nonlinear MPC is formulated with a horizon of~$N=180$, the input constraint is $\mathcal{U}=[0,2]$,
and we obtain a cubic domain~$\mathcal{X}=[-0.2,0.2]^2$ after shifting, as done in~\cite{Hertneck2018Learning}.\footnote{%
Before executing \alkiax, we shift and scale the domain~$\domain$ to yield the unit cube~$[0,1]^2$.
}
Figure~\ref{fig:ground_truth} shows the corresponding MPC feedback law~$f$ over the domain $\domain=[-0.2,0.2]^2$.

We use \alkiax\ to determine an approximate MPC to control system~\eqref{eq:CSTR} with high sampling rates and closed-loop guarantees.
To this end, we consider the MPC scheme in~\cite{Hertneck2018Learning}, which is designed to be robust with respect to input disturbances bounded by $\error=5.1\cdot 10^{-3}$.
To compute the approximate MPC, \alkiax\ requires samples, \ie solutions of the NLP of the ground truth MPC~$f$ for different states.
In this experiment, the hyperparameters were $\underline p=2$ and $\overline p=5$, yielding an upper bound on the condition number of~$\overline\kappa = 1.14\cdot 10^8$ (see~\eqref{eq:condition}, \eqref{eq:condition_extended}).

\subsubsection*{Result}\label{sec:result}
We approximate the ground truth MPC feedback law~$f$ using \alkiax\ with the error bound~$\error=5.1\cdot 10^{-3}$.
We heuristically validate the approximation error by evaluating the error for $90\cdot 10^3$ equidistant states on the domain~$\domain$, yielding a maximum error of $2.46\cdot 10^{-3}$, \ie  about 50\% smaller than the guaranteed error bound~$\error$.

\begin{figure}
    \centering
    \input{figures_final/subdomains_5e-3}
    \caption{Sub-domain partitioning using \alkiax\ and feasible domains.
    The light gray shaded area corresponds to the feasible domain $\xfeasf$ of the ground truth MPC, whereas the area enclosed by the magenta shading corresponds to the feasible domain~$\xfeash$ of the obtained approximate MPC via \alkiax.
    The resulting partitioning is reminiscent of classical linear explicit MPC approaches~\cite{Johansen2003Tree,Summers2011Multiresolution}.
    In total, we have~1996 sub-domains.
    }
    \label{fig:subdomains}
\end{figure}
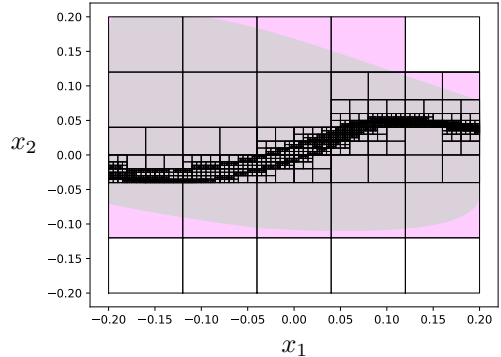

Figure~\ref{fig:subdomains} depicts the
sub-domain partitioning of \alkiax, the feasible domain~$\xfeasf$ of the MPC, and~$\xfeash$, 
the feasible domain computed by \alkiax.
Recall that \alkiax\ partitions sub-domains further if the maximum number of samples cannot ensure the desired approximation error~$\error$.
This behavior can be observed in Figure~\ref{fig:subdomains}, where larger sub-domains correspond to regions where the ground truth 
resembles a constant function
(see Figure~\ref{fig:ground_truth}), and hence fewer samples are required.
Furthermore, the derived sample complexities in Theorem~\ref{th:alkiax} and Corollary~\ref{co:SE} also support this statement.\footnote{%
The sample complexity on sub-domain~$\domain_a$ increases with the RKHS norm~$\|\tilde{f}_a\|_\kla$ 
(see Theorem~\ref{th:alkiax}, Corollary~\ref{co:SE}).
For constant functions, we have~$\|\tilde{f}_a\|_\kla=0$.}
Overall, \alkiax\ successfully approximates the ground truth MPC~$f$, yielding an approximate MPC~$h$ with closed-loop guarantees on stability and constraint satisfaction.

\subsubsection*{Discussion}
In the following, we compare the MPC approximation obtained using \alkiax\ with the approach presented in~\cite{Hertneck2018Learning}.
Both works approximate the same ground truth MPC (see~\eqref{eq:CSTR}, Figure~\ref{fig:ground_truth}) with the identical error bound~$\error$.
\alkiax\ is a kernel-based algorithm that inherently guarantees the satisfaction of any desired error bound~$\error>0$  (see Theorem~\ref{th:alkiax}, \ref{property:error_satisfaction}, \ref{property:error_samples}).
In contrast, the approach in~\cite{Hertneck2018Learning} uses NNs for the approximation and an additional statistical validation step to guarantee the desired approximation accuracy.
Notably, the approach in~\cite{Hertneck2018Learning} may lead to an \emph{iterative} offline design between NN training and validation, \eg
requiring additional sampling or changes to the NN architecture whenever the statistical validation step fails.

Table~\ref{tab:experiment} compares the number of samples $\mathrm{card}(X)$, the offline approximation time $t_\mathrm{offline}$, and the online evaluation time $t_\mathrm{online}$ of \alkiax\ and the NN approach proposed in~\cite{Hertneck2018Learning}.\footnote{%
Table~\ref{tab:experiment} directly contains the computation times reported in~\cite{Hertneck2018Learning}, which are not obtained with the same hardware as \alkiax.
}
The online evaluation time of \alkiax\ is obtained over three runs of evaluating~$90\cdot 10^3$ equidistant states.
%
%
%
\begin{table}
    \centering
    \begin{tabular}{c | c|c}
\hline
         & \alkiax & NN approach~\cite{Hertneck2018Learning}  \\
        \hline
        \rule{0pt}{2.25ex} 
     $t_\text{online}$ & $\SI{44.02}{\micro\second} \pm \SI{0.14}{\micro\second}$ & \SI{3000}{\micro\second} \\
     $t_\text{offline}$ & \SI{10.3}{\hour} & \SI{500}{\hour} \\
         $\mathrm{card}(X)$ & 
     $1.56\cdot 10^{6}$ & 
     $1.6\cdot 10^6$ \\
         \hline
    \end{tabular}
    \caption{Comparison between \alkiax\ and the NN approximation in~\cite{Hertneck2018Learning} with the continuous stirred tank reactor.
    }
    \label{tab:experiment}
\end{table}
Both \alkiax\ and the NN approach require roughly~$\mathrm{card}(X)\approx10^6$ samples. 
However, the offline approximation time~$t_\mathrm{offline}$ of~\cite{Hertneck2018Learning} is significantly longer. 
The main reason for this discrepancy is that the NN approach requires further samples for the statistical validation step.
Additionally, \alkiax\ yields an almost 100-times faster\footnote{%
More recently,~\cite{Hose2023Approximate} approximated the same ground truth MPC~$f$ with NNs, achieving an online evaluation time of $t_\mathrm{online}=\SI{500}{\micro\second}$, 
which is still 10-times slower than the proposed solution.
Furthermore, additional parallelization with $10^3$ cores and the usage of GPUs reduced the offline sampling time to $t_\mathrm{offline}=\SI{0.5}{\hour}$, resulting in $t_\mathrm{offline}=\SI{4.5}{\hour}$ overall including the NN training.
However, in contrast to the proposed approximation method,~\cite{Hose2023Approximate} does not guarantee error bound~$\error$ but utilizes a fallback policy to obtain closed-loop guarantees.
}
online evaluation time $t_\mathrm{online}$ compared to the NN approach.
The fast online evaluation of \alkiax~\ref{property:fast} is due to the \emph{localized} kernel interpolation approach~\ref{tool:localized}, whereas standard kernel-based methods tend to  suffer from scalability issues~\cite[Section~2.3]{Rasmussen2006Gaussian}.

In addition to the presented numerical differences, the closed-loop guarantees on the approximate MPC using \alkiax\ are \emph{deterministic}.
Furthermore, both the sample acquisition and the offline approximation of \alkiax\ are entirely parallelized, thus accelerating the offline approximation.
Besides the favorable offline approximation time, 
the offline approximation with \alkiax\ is an \emph{automatic} and \emph{non-iterative} process~\ref{property:numerical} 
in contrast to the iterative offline design present in~\cite{Hertneck2018Learning}.



\subsection{Cold atmospheric plasma} \label{sec:plasma}
We consider the MPC scheme to control a cold atmospheric plasma device from~\cite[Chapter~4.5]{Bonzanini2022Thesis}. 
The nonlinear systems dynamics $g{:}\; \mathbb{R}^3 \rightarrow \mathbb{R}^2$ are described by
\begin{align} \label{eq:plasma}
    g(x,u) = \begin{pmatrix}
  0.427 x_1 + 0.68  x_2 + 1.58  u_1 - 1.02  u_2      \\
-0.06  x_1 + 0.26  x_2 + 0.73  u_1 + 0.03  u_2        \\
x_3 + \mathbb{K}^{(\mathbb T-x_1)}\delta
    \end{pmatrix},
\end{align}
where~$x_1$ is the surface temperature,~$x_2$ is the gas temperature, and~$x_3$ is the thermal dose delivered to the target surface.
The inputs~$u_1$ and~$u_2$ correspond to the applied power to the plasma and the flow rate of Argon, respectively, while~$\mathbb K$ and~$\mathbb T$ are constants and~$\delta=\SI{0.5}{\second}$ is the sampling time.
Note that states and inputs are appropriately shifted in~\eqref{eq:plasma}, see~\cite[Chapter~4.5]{Bonzanini2022Thesis}.
The state constraints%
\footnotemark[8]
are~$x_1\in[\SI{25}{\degreeCelsius} , \SI{42.5}{\degreeCelsius}]$,~$x_2\in[\SI{20}{\degreeCelsius} , \SI{80}{\degreeCelsius}]$, and~$x_3\in[\SI{0}{\minute} , \SI{11}{\minute}]$, whereas the input constraints 
are~$u_1\in[\SI{1.5}{\watt}, \SI{8}{\watt}]$
and~$u_2\in[\SI{1}{slm}, \SI{6}{slm}]$.
The control objective is to achieve a treatment time of the thermal dose of~$\overline x_{3}=\SI{10}{\minute}$ and satisfy the state and input constraints.
The nonlinear MPC is formulated with a horizon of~$N=120$.

As mentioned in~\cite[Sections~1.2.2 and 6.2.2]{Bonzanini2022Thesis}, approximate MPC schemes are desired for this application due the need for \emph{(i)}~sub-second sampling rates and \emph{(ii)}~implementability on portable embedded devices.
To this end, we use \alkiax\ to determine an approximate MPC to control system~\eqref{eq:plasma}.
In this application example, we focus on obtaining an implementable controller for relatively cheap hardware with limited memory capacities.
Hence, instead of ensuring a desired error bound~$\error$, we here focus on achieving the best approximation subject to a memory requirement on the samples of \SI{75}{\mega\byte}, which we achieve by upper-bounding the maximum depth~$\overline d$ (see~\eqref{eq:depth}) a priori.
Specifically, for this application, we pre-partitioned the domain into $3^3$ sub-domains and obtained~$\overline d=1$, \ie each sub-domain can at most be partitioned into $2^3$ further sub-domains.
The hyperparameters for this experiment were $\underline p=2$ and $\overline p=4$, yielding an upper bound on the condition number of~$\overline\kappa = 3.32\cdot 10^7$ (see~\eqref{eq:condition}, \eqref{eq:condition_extended}).

\subsubsection*{Result}
\begin{table}
    \centering
    \begin{tabular}{c | c|c}
\hline
         & Approximate MPC (\alkiax) & MPC (IPOPT) \\
        \hline
        \rule{0pt}{2.25ex} 
     $t_\text{online}$ & $\SI{96}{\mu \second}\pm\SI{4}{\micro\second}$ & $\SI{507}{\milli \second}\pm\SI{6}{\milli\second}$ \\
     Offline memory & \SI{33}{\mega\byte} & \SI{9}{\mega\byte} \\
     $t_\text{offline}$ & \SI{69}{\hour} & [-]\\
     $\mathrm{card}(X)$ & $8.22\cdot 10^5$ & [-] \\
         \hline
    \end{tabular}
    \caption{Comparison between the approximate MPC computed by \alkiax\ and the MPC with the cold atmospheric plasma device.
    }
    \label{tab:plasma}
\end{table}
Table~\ref{tab:plasma} compares the approximate MPC computed by \alkiax\ (Algorithm~\ref{alg:alkiax_offline}) and the MPC using online optimization (IPOPT) in terms of online evaluation time~$t_\mathrm{online}$, offline memory requirement, offline approximation time~$t_\mathrm{offline}$, and number of samples~$\mathrm{card}(X)$.
The documented online evaluation times are obtained over three closed-loop runs of 30 time steps from the initial state~$x=[\SI{35}{\degreeCelsius}, \SI{58}{\degreeCelsius}, \SI{0}{\min}]^\top$ and the NLP of the MPC is solved using a warm start.
\alkiax\ requires a long time $t_{\mathrm{ofline}}$ for the offline approximation\footnote{%
Note that, due to parallelization, the offline approximation time can directly be reduced to $t_\mathrm{offline}\approx\SI{2.5}{\hour}$ by increasing the number of cores from~8 to~216, provided sufficient RAM.}, 
whereas the MPC only runs the just-in-time compilation prior to deployment. 
The offline memory requirements for the stored samples while executing \alkiax\ and for the compiled IPOPT code are comparable, both in the order of tens of megabytes.
The benefit of the approximate MPC is the significant speedup of the online evaluation time~$t_\mathrm{online}$ by a factor of over~$10^3$.
%

\begin{figure}
    \centering
\input{figures_final/x1_plasma.tex}
\input{figures_final/x2_plasma.tex}
\input{figures_final/x3_plasma.tex}
\input{figures_final/u1_plasma.tex}
\input{figures_final/u2_plasma.tex}
\caption{Closed-loop state and input trajectories of the MPC (blue lines) and the approximate MPC (magenta lines) computed by \alkiax.
The dotted black line shows the constraint on~$x_1$, whereas the dashed green line depicts the desired treatment time of~$\overline x_3=\SI{10}{\min}$.
}
\label{fig:plasma}
\end{figure}
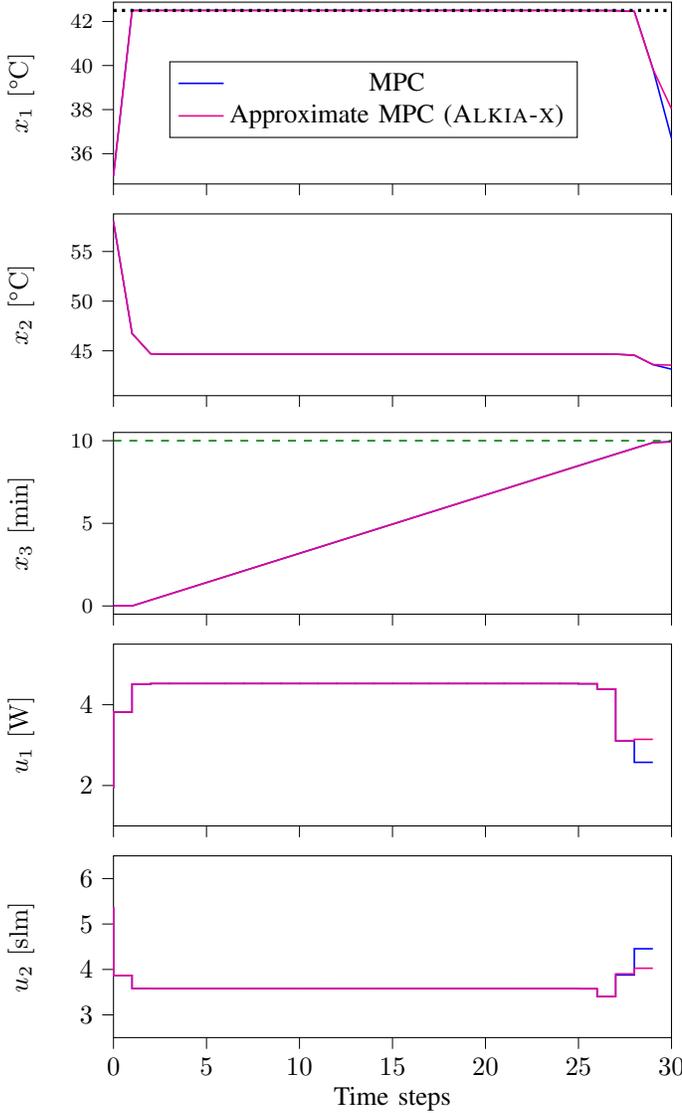

Figure~\ref{fig:plasma} illustrates the closed-loop trajectories of the MPC, where the NLP is solved using IPOPT, and of the approximate MPC computed by \alkiax. 
As in~\cite[Chapter~4.5]{Bonzanini2022Thesis}, closed-loop simulations are terminated as soon as $x_3\geq \overline{x}_3$.
The trajectories are very similar, with~$x_3$ showing no noticeable visual difference and all constraints are satisfied.
Particularly, the approximation error between the approximate MPC and the MPC on the state trajectory of the approximate MPC is~$\max_t\|f(x_t)-h(x_t)\|_\infty = 2.98\cdot 10^{-3}$.

 Overall, \alkiax\ \emph{(i)} automatically computes an approximate MPC~\ref{property:numerical} that closely resembles the MPC, \emph{(ii)} reduces the online evaluation time by a factor of over $10^3$~\ref{property:fast}, and \emph{(iii)} only requires~\SI{33}{\mega\byte} of memory, thus being implementable on most microcontrollers.


%% file: figures_final/ground_truth.tex
\begin{tikzpicture}[xscale=1, yscale=1]
\centering
\pgfplotsset{
every axis legend/.append style={
at={(0.5,0.5)},
anchor=north west,
},
}

\node[inner sep=0pt] (whitehead) at (0, 0)
    {\includegraphics[width=8cm]{figures_final/CSTR_downsized.jpg}};
\node[align=left] at (-3, 0.2){{$f(x)$}};
\node[align=left] at (-2.1, -2){{$x_2$}};
\node[align=left] at (1.8, -2.2){{$x_1$}};


\end{tikzpicture}

%% file: figures_final/subdomains_5e-3.tex
\begin{tikzpicture}[xscale=1, yscale=1]
\centering
\pgfplotsset{
every axis legend/.append style={
at={(0.5,0.5)},
anchor=north west,
},
}

\node[inner sep=0pt] (whitehead) at (0, 0)
    {\includegraphics[width=7cm]{figures_final/domains_CSTR.pdf}};
\node[align=left] at (-3.5, 0.1){{$x_2$}};
\node[align=left] at (0.1, -2.6){{$x_1$}};


\end{tikzpicture}

%% file: figures_final/x1_plasma.tex
\begin{tikzpicture}

\definecolor{darkgray176}{RGB}{176,176,176}
\pgfplotsset{
every axis legend/.append style={
at={(0.1,0.675)},
anchor=north west,
},
}
\begin{axis}[
tick align=outside,
tick pos=left,
x grid style={darkgray176},
xmin=-1.5, xmax=31.5,
xtick style={color=black},
y grid style={darkgray176},
tick label style={font=\footnotesize},
xticklabels={},
ylabel={$x_1$ \SI{}{[\degreeCelsius]}},
ymin=34.6250011057457, ymax=42.8749767793399,
ytick style={color=black},
width=9cm,
height=4cm,
xmin=0, xmax=30,
]
\addplot [semithick, blue]
table {%
0 35
1 42.4999625324934
2 42.499959678823
3 42.4999564900159
4 42.4999529078333
5 42.4999488645027
6 42.4999442772897
7 42.4999390440703
8 42.4999330375901
9 42.4999260976553
10 42.4999180203432
11 42.4999085428539
12 42.4998973218703
13 42.4999578636321
14 42.4999519692184
15 42.499944745615
16 42.4999357609468
17 42.4999243925561
18 42.4997513446633
19 42.4996979575905
20 42.4996253417503
21 42.4995230248604
22 42.4993723040216
23 42.499137068782
24 42.4987400379286
25 42.497990405123
26 42.4963048278321
27 42.4911236749312
28 42.4571043772555
29 39.8621049920065
30 36.7004986410263
};
\addplot [semithick, magenta]
table {%
0 35
1 42.4935171411996
2 42.4952485547761
3 42.4982259373822
4 42.4985152336006
5 42.4983411022883
6 42.4982007589267
7 42.4981111343009
8 42.4980730257484
9 42.4980863021829
10 42.4981505490593
11 42.4982650327144
12 42.4984347332417
13 42.4986475236695
14 42.4989935991688
15 42.4988203433113
16 42.4985463843662
17 42.4983182317686
18 42.4980293131561
19 42.4978810812188
20 42.4977665320423
21 42.4977312769792
22 42.4975865673357
23 42.4974810658922
24 42.4972473983821
25 42.4974002398475
26 42.4965707489379
27 42.4904798401375
28 42.4591528234963
29 39.8484771725129
30 38.034035632099
};
\addplot [dotted, very thick, black] 
table {%
0 42.5
1 42.5
2 42.5
3 42.5
4 42.5
5 42.5
6 42.5
7 42.5
8 42.5
9 42.5
10 42.5
11 42.5
12 42.5
13 42.5
14 42.5
15 42.5
16 42.5
17 42.5
18 42.5
19 42.5
20 42.5
21 42.5
22 42.5
23 42.5
24 42.5
25 42.5
26 42.5
27 42.5
28 42.5
29 42.5
30 42.5
};
\legend{
MPC,
Approximate MPC (\alkiax)
}
\end{axis}

\end{tikzpicture}

%% file: figures_final/x2_plasma.tex
\begin{tikzpicture}

\definecolor{darkgray176}{RGB}{176,176,176}

\begin{axis}[
tick align=outside,
tick pos=left,
x grid style={darkgray176},
xmin=-1.5, xmax=31.5,
xtick style={color=black},
y grid style={darkgray176},
ymin=40.4836578632328, ymax=58.7864924827032,
ytick style={color=black},
width=9cm,
height=4cm,
xmin=0, xmax=30,
ytick style={color=black},
tick label style={font=\footnotesize},
xticklabels={},
ylabel={$x_2$ \SI{}{[\degreeCelsius]}}
]
\addplot [semithick, blue]
table {%
0 58
1 46.7417278078668
2 44.6958352000065
3 44.6626959240632
4 44.6622863000222
5 44.662280487092
6 44.6622795360705
7 44.6622785147558
8 44.6622773355478
9 44.6622759631863
10 44.6622743532643
11 44.6622724479168
12 44.662270170619
13 44.6622952082721
14 44.6622806478149
15 44.6622789653632
16 44.6622770473292
17 44.662274582844
18 44.6622115235396
19 44.662226508222
20 44.6622100383453
21 44.6621856502217
22 44.6621483979589
23 44.6620875884294
24 44.6619788060248
25 44.6617558601323
26 44.6611686621309
27 44.6576831984824
28 44.5573390002871
29 43.6111387278409
30 43.1500395275924
};
\addplot [semithick, magenta]
table {%
0 58
1 46.7402415055556
2 44.6956829645966
3 44.6631573133605
4 44.6622397691361
5 44.6621069604931
6 44.6620923805377
7 44.6620899528022
8 44.6620884394629
9 44.6621009246051
10 44.6621208816369
11 44.6621480355227
12 44.6621843887848
13 44.6622243740982
14 44.662304310296
15 44.6621714208121
16 44.6621085174057
17 44.6620835414547
18 44.6620252026147
19 44.6620308685176
20 44.6620192550061
21 44.6620288720094
22 44.6619788525183
23 44.6619765241067
24 44.6618983260404
25 44.6619886973262
26 44.6615681003361
27 44.6575828357941
28 44.5596108402049
29 43.6162972291031
30 43.5544196823137
};
\end{axis}

\end{tikzpicture}

%% file: figures_final/x3_plasma.tex
\begin{tikzpicture}

\definecolor{mydarkgreen}{RGB}{0,125,0}
\begin{axis}[
xmin=-1.5, xmax=31.5,
tick align=outside,
tick pos=left,
x grid style={darkgray176},
xtick style={color=black},
y grid style={darkgray176},
ymin=-0.499986756032121, ymax=10.4997218766746,
ytick style={color=black},
width=9cm,
height=4cm,
xmin=0, xmax=30,
ytick style={color=black},
xticklabels={},
tick label style={font=\footnotesize},
ylabel={$x_3$ \SI{}{[\minute]}}
]
\addplot [semithick, blue]
table {%
0 0
1 0.001953125
2 0.355497333755398
3 0.709040843196216
4 1.06258357119622
5 1.41612542135787
6 1.76966628067632
7 2.12320601587318
8 2.47674446864541
9 2.8302814495076
10 3.18381672972094
11 3.53735003057847
12 3.89088100897908
13 4.24440923769955
14 4.59795230231417
15 4.95149392246219
16 5.30503377242455
17 5.65857142065454
18 6.01210628303018
19 6.36559874226982
20 6.71907812072804
21 7.07253970779215
22 7.42597622801478
23 7.77937582604335
24 8.13271780603089
25 8.48596255939603
26 8.83902381239337
27 9.1916728060439
28 9.54305760224151
29 9.88625353897135
30 9.94305630491534
};
\addplot [semithick, magenta]
table {%
0 0
1 0.001953125
2 0.353921360983386
3 0.706312256219609
4 1.05943115418047
5 1.41262086835822
6 1.76576795559453
7 2.11888069084557
8 2.47197149033552
9 2.82505296311319
10 3.17813768514622
11 3.53123813128932
12 3.88436659848504
13 4.23753660572213
14 4.59075870764125
15 4.9440655510856
16 5.29732996781953
17 5.65052730817498
18 6.00366879714337
19 6.35673957197572
20 6.70977407186786
21 7.06278054213165
22 7.41577838610021
23 7.76874082436702
24 8.12167745213047
25 8.47455692079766
26 8.8274737760716
27 9.18018777684478
28 9.53141579473655
29 9.87509937281948
30 9.93136810128007
};
\addplot [semithick, dashed, mydarkgreen]
table {%
0 10
1 10
2 10
3 10
4 10
5 10
6 10
7 10
8 10
9 10
10 10
11 10
12 10
13 10
14 10
15 10
16 10
17 10
18 10
19 10
20 10
21 10
22 10
23 10
24 10
25 10
26 10
27 10
28 10
29 10
30 10
};
\end{axis}

\end{tikzpicture}

%% file: figures_final/u1_plasma.tex
\begin{tikzpicture}

\definecolor{darkgray176}{RGB}{176,176,176}

\begin{axis}[
xmin=-1.45, xmax=30.45,
tick align=outside,
tick pos=left,
x grid style={darkgray176},
xtick style={color=black},
y grid style={darkgray176},
ymin=1, ymax=5.5,
xmin=0, xmax=30,
ytick style={color=black},
ylabel={$u_1$ \SI{}{[W]}},
xticklabels={},
width=9cm,
height=4cm
]
\addplot [semithick, blue, const plot mark right]
table {%
0 1.93252681818045
1 3.81769117616636
2 4.51277278347709
3 4.52419512974747
4 4.52433498333745
5 4.52433543284338
6 4.52433398074189
7 4.52433228170965
8 4.52433030804817
9 4.52432799772715
10 4.52432526976918
11 4.524322017355
12 4.5243566887129
13 4.52433239966889
14 4.5243348481402
15 4.52433220105223
16 4.52432873679592
17 4.52424116031345
18 4.52427055724986
19 4.52423791585279
20 4.52420411443129
21 4.52415294487833
22 4.52406990376238
23 4.52392225955871
24 4.52362160142839
25 4.52283590939858
26 4.51825911291421
27 4.3886782837444
28 3.10588877072037
29 2.57425802943837
};
\addplot [semithick, magenta, const plot mark right]
table {%
0 1.93036849735187
1 3.81745712071652
2 4.51306725103855
3 4.52380817346947
4 4.52397916042848
5 4.52399227344775
6 4.52398239758785
7 4.52397438084766
8 4.5239891703327
9 4.52401361863667
10 4.52404967184766
11 4.52410015177072
12 4.52415701248591
13 4.52427170818802
14 4.52408641089236
15 4.52403231084067
16 4.52399737493183
17 4.52390630389662
18 4.52391129428938
19 4.52388085930057
20 4.52388906816477
21 4.52381353636872
22 4.52381569881174
23 4.52370050382594
24 4.52383654054592
25 4.52324023296867
26 4.51797774830367
27 4.39167841419416
28 3.11132255065645
29 3.14286697837784
};
\end{axis}

\end{tikzpicture}

%% file: figures_final/u2_plasma.tex
\begin{tikzpicture}

\definecolor{darkgray176}{RGB}{176,176,176}

\begin{axis}[
xmin=-1.45, xmax=30.45,
tick align=outside,
tick pos=left,
x grid style={darkgray176},
xtick style={color=black},
y grid style={darkgray176},
ymin=2.5, ymax=6.5,
ytick style={color=black},
ylabel={$u_2$ \SI{}{[slm]}},
xlabel={Time steps},
width=9cm,
xmin=0, xmax=30,
height=4cm
]
\addplot [semithick, blue, const plot mark right]
table {%
0 5.36610768650165
1 3.86897211031035
2 3.58174069508624
3 3.57734348170326
4 3.57728952399565
5 3.57728917736912
6 3.57728953577601
7 3.57728995691632
8 3.57729044413833
9 3.57729101181141
10 3.5772916785738
11 3.57729246875759
12 3.57728068209876
13 3.57729045745974
14 3.57728919807097
15 3.57728981011866
16 3.57729061114019
17 3.57731828415633
18 3.57730286630374
19 3.57731150301572
20 3.57731857406954
21 3.57732868772378
22 3.57734378182717
23 3.5773669227398
24 3.57740012662814
25 3.57737829809737
26 3.57428278569778
27 3.40245473480378
28 3.87860159252368
29 4.45537987947142
};
\addplot [semithick, magenta, const plot mark right]
table {%
0 5.36908341629051
1 3.86958345003693
2 3.58185208646114
3 3.57774857575181
4 3.57769158069951
5 3.57768924433058
6 3.5776943052968
7 3.57768072578706
8 3.57767391822738
9 3.57766259219114
10 3.57764595972827
11 3.57762302397212
12 3.57759659646412
13 3.57754924911393
14 3.57762787112824
15 3.57765272265204
16 3.57766754303267
17 3.57769912993396
18 3.57769432653131
19 3.57770222583967
20 3.57769459566967
21 3.57771136218093
22 3.57772521198409
23 3.57773086444131
24 3.57764339533623
25 3.57765611183372
26 3.57485394604386
27 3.40476169816492
28 3.90273726934514
29 4.02660462939559
};
\end{axis}

\end{tikzpicture}

%% file: Sections/summary_outlook.tex
\section{Conclusion} \label{sec:conclusion}
We addressed the problem of automatically approximating nonlinear MPC schemes while ensuring that the resulting approximate MPC inherits desirable closed-loop properties on stability and constraint satisfaction.
By considering a robust MPC scheme, we reduced this problem to a function approximation problem, which we tackled by proposing \alkiax.
\alkiax\ is an automatic and reliable algorithm that yields a computationally efficient approximate MPC with closed-loop guarantees.
We successfully applied \alkiax\ to approximate two nonlinear MPC schemes, \emph{(i)}~demonstrating reduced offline computation and online evaluation time by over one order of magnitude and \emph{(ii)}~showcasing applicability to realistic problems with limited memory capacities.

Although we proposed \alkiax\ to approximate MPC schemes, the algorithm is equally capable of automatically approximating a wide range of black-box functions with guaranteed bounds on the approximation error.

%% file: Sections/acknowledgements.tex
\section*{Acknowledgements}
The authors thank D. Baumann for helpful comments and J. Sieber for setting up the server.

%% file: Sections/appendix.tex
\appendices


\section{Proof of Lemma~\ref{le:cont}} 
\label{app:cont}
Analogous to~\cite[Proof of Lemma~4.23]{Steinwart2008SVM}, for all $x, x^\prime \in \domain$, it holds that
\begin{align*}
&\lvert f(x) - f(x^\prime) \rvert = \lvert \langle f, k(x, \cdot) - k(x^\prime, \cdot) \rangle_k \rvert \\
&\leq \|f\|_k \sqrt{k(x,x) - k(x^\prime, x) - k(x, x^\prime) + k(x^\prime, x^\prime)} \\
\overset{\mathrm{Asm.\ref{asm:kernel}}}&{=} \|f\|_k \sqrt{2\left(1-\widetilde k(\|x-x^\prime\|_2)\right)} 
\overset{\eqref{eq:knorm}}{=}\|f\|_k \sqrt{2\knorm(\|x-x^\prime\|_2)},
\end{align*}
where $\langle f, g \rangle_k$ denotes the inner product between two functions~$f,g$ in the RKHS of kernel~$k$.
The first equality follows from the reproducing property~\cite[Definition~4.18]{Steinwart2008SVM},
while the first inequality follows from Cauchy-Schwarz.\hfill\qedsymbol

\section{Proof of Lemma~\ref{le:constants}} \label{sec:appendix_lemma}
In Part~I, we show that, for any $a\in\mathcal A$,~\eqref{eq:extrapolation} interpolates~$\hat \gamma_{a, p}$ for $ p \in \{\underline p, \underline p+1\}$,
\ie~\eqref{eq:interpolating} holds.
In Part~II, we prove that $\rkhsx \in (0,\infty)$ holds for any $a \in \mathcal A$.\\
\textit{Part~I:}
Consider any $a \in \mathcal{A}$.
First, we show $\gamma_{a}(\underline p) = \hat{\gamma}_{a,\underline p}$.
It holds that
\begin{align*}
    &\gamma_{a}(\underline p) \overset{\eqref{eq:extrapolation}}{=}\rkhso \exp\left(\frac{-\tau_{a}}{1+2^{\underline p}}\right) \\
    \overset{\eqref{eq:gamma_bar}}&{=}
     \hat{\gamma}_{a, \underline p} 
     \left( \frac{\hat{\gamma}_{a, \underline p+1}}{\hat{\gamma}_{a, \underline p}}
        \right)
    ^{
    \lambda
    }
    \exp\left(\frac{
    -
    \ln\left( \frac{\hat{\gamma}_{a, \underline p+1}}{\hat{\gamma}_{a, \underline p}}
        \right) \lambda \left( 1+ 2^{\underline p}
        \right)
    }
    {1+2^{\underline p}}\right) \\
    &=
    \hat{\gamma}_{a, \underline p} \left(\frac{\hat{\gamma}_{a, \underline p+1}}{\hat{\gamma}_{a, \underline p}}\right)^\lambda \exp\left(-\ln\left(\frac{\hat{\gamma}_{a, \underline p+1}}{\hat{\gamma}_{a, \underline p}}\right)\lambda\right) \\
    &=
    \hat{\gamma}_{a, \underline p} \left(\frac{\hat{\gamma}_{a, \underline p+1}}{\hat{\gamma}_{a, \underline p}}\right)^\lambda \left(\frac{\hat{\gamma}_{a, \underline p}}{\hat{\gamma}_{a, \underline p+1}}\right)^\lambda = \hat{\gamma}_{a, \underline p}.
\end{align*}
Second, we show $\gamma_{a}(\underline p +1) = \hat{\gamma}_{a,\underline p+1}$.
It holds that
\begin{align*}
    &\gamma_{a}(\underline p+1) \overset{\eqref{eq:extrapolation}}{=} \rkhso \exp\left(\frac{-\tau_{a}}{1+2^{\underline p+1}}\right) \\
    \overset{\eqref{eq:gamma_bar}}&{=}
     \hat{\gamma}_{a, \underline p} 
     \left( \frac{\hat{\gamma}_{a, \underline p+1}}{\hat{\gamma}_{a, \underline p}}
        \right)
    ^{
    \lambda
    }
    \exp\left(\frac{
    -
    \ln\left( \frac{\hat{\gamma}_{a, \underline p+1}}{\hat{\gamma}_{a, \underline p}}
        \right) \lambda \left( 1+ 2^{\underline p}
        \right)
    }
    {1+2^{\underline p+1}}\right) \\
    \overset{\eqref{eq:gamma_hat}}&{=}
     \hat{\gamma}_{a, \underline p} 
     \left( \frac{\hat{\gamma}_{a, \underline p+1}}{\hat{\gamma}_{a, \underline p}}
        \right)
    ^{
    \lambda
    }
    \exp\left(\frac{
    -
    \ln\left( \frac{\hat{\gamma}_{a, \underline p+1}}{\hat{\gamma}_{a, \underline p}}
        \right) \lambda \left( 1+ 2^{\underline p}
        \right)
    }
    {2^{\underline p} \lambda}\right)\\    
        &=
     \hat{\gamma}_{a, \underline p} 
     \left( \frac{\hat{\gamma}_{a, \underline p+1}}{\hat{\gamma}_{a, \underline p}}
        \right)
    ^{
    \lambda
    }
    \exp\left(
    -
    \ln\left( \frac{\hat{\gamma}_{a, \underline p+1}}{\hat{\gamma}_{a, \underline p}}
        \right)  \left( 1+ 2^{-\underline p}
        \right)
    \right)\\    
    \overset{\eqref{eq:gamma_hat}}&{=}
         \hat{\gamma}_{a, \underline p} 
     \left( \frac{\hat{\gamma}_{a, \underline p+1}}{\hat{\gamma}_{a, \underline p}}
        \right)
    ^{
    \lambda
    }
    \exp\left(
    -
    \ln\left( \frac{\hat{\gamma}_{a, \underline p+1}}{\hat{\gamma}_{a, \underline p}}
        \right) \left(\lambda-1\right)\right) \\
    &=  \hat{\gamma}_{a, \underline p}  
    \left( \frac{\hat{\gamma}_{a, \underline p+1}}{\hat{\gamma}_{a, \underline p}}
        \right)
    ^{
    \lambda
    }
    \left( \frac{\hat{\gamma}_{a, \underline p}}{\hat{\gamma}_{a, \underline p+1}}
        \right)
    ^{
    \lambda -1
    } \\
        &=
    \hat{\gamma}_{a, \underline p}  
    \left( \frac{\hat{\gamma}_{a, \underline p+1}}{\hat{\gamma}_{a, \underline p}}
        \right)
    = \hat{\gamma}_{a, \underline p+1},
\end{align*}
\ie \eqref{eq:interpolating} holds.\\
\textit{Part~II:} 
From~\eqref{eq:extrapolation}, we have that $\rkhsx \in (0,\infty)$ if $\hat\gamma_{a, p} > 0, \; p \in \{\underline p, \underline p+1\}$ and $\hat\gamma_{a, \underline p+1} < \infty$.
The fact that $\hat\gamma_{a, p} > 0, \; p \in \{\underline p, \underline p+1\}$ directly follows from~\eqref{eq:gamma_hat} with $\error>0$.
Moreover, from~\eqref{eq:extrapolation}, it follows that $\hat\gamma_{a, \underline p+1} < \infty$ holds if $\|\widetilde h_{X_{a,\underline p+1}}\|_\kla < \infty$ with $\error < \infty$.
Since the covariance matrix is always strictly positive definite (see  Section~\ref{sec:kernel_interpolation}) and we have a continuous ground truth $f$ on a bounded set $\domain$ (see Lemma~\ref{le:cont}), $\|\widetilde h_{X_{a,\underline p+1}}\|_\kla < \infty$ follows from~\eqref{eq:RKHS_h_a}.
\hfill \qedsymbol

\section{Proof of Lemma~\ref{le:SE}} \label{sec:appendix_lemma_SE_limit}

%
For the SE-kernel $\widetilde k_\mathrm{SE}{:}\; \mathbb{R}_{\geq 0} \rightarrow [0,1)$ with $\widetilde k_\mathrm{SE}(x) = \exp\left(-x^2\right)$ \cite{Rasmussen2006Gaussian}, we have  $\knorm_\mathrm{SE}(x) = 1-\exp\left(-x^2\right)$ (see~\eqref{eq:knorm}) and $\kinv_\mathrm{SE}{:}\; [0,1) \rightarrow  \mathbb{R}_{\geq 0} $ with
\begin{align} \label{eq:k_SE}
    \kinv_\mathrm{SE}(x) = \sqrt{\ln\left(\frac{1}{1-x}\right)}.
\end{align}
With L'Hôpital's rule, it follows that 
\begin{align}\label{eq:hospital}
     \frac{\ln\left(
    \frac{1}{1-x}
    \right)}{x} 
    = \frac{1}{1-x} = 1,
    \quad \mathrm{as}\; x \rightarrow 1.
\end{align}
Hence,~\eqref{eq:k_SE} and~\eqref{eq:hospital} yield
\begin{align}
 \kinv_\mathrm{SE}(x) = \sqrt{x},
    \quad \mathrm{as}\; x \rightarrow 1.
\end{align}
Thus, with~\eqref{eq:ell_a_min}, it holds for any $a \in \domain_a$ that
\begin{align*}
    %
    \ell_a 
    \leq 
    \mathcal{O}
    \left(
    \frac{\sqrt{C} \error}{\sqrt{n}\|f\|_k}
    \right)
    \quad
\mathrm{as}\;\frac{\error}{\|f\|_k}\rightarrow 0,
\end{align*}
which yields~\eqref{eq:ell_limit}.
\hfill \qedsymbol

%% file: Sections/biographies.tex
\begin{IEEEbiography}[{\includegraphics[width=1in,height=1.25in,clip,keepaspectratio]{Tokmak_2023_crop.JPG}}]{Abdullah Tokmak} received his B.Sc. and M.Sc. degrees in mechanical engineering from RWTH Aachen University, Aachen, Germany, in 2020 and 2022, respectively.
In 2022, he received the Springorium Commemorative Coin for his Master's degree and the Friedrich-Wilhelm Award for his Master's thesis, both from RWTH Aachen University.
He is currently a doctoral researcher at the Cyber-physical Systems Group of  Aalto University, Espoo, Finland. 
His research interests include machine learning and control.
\end{IEEEbiography}

\begin{IEEEbiography}[{\includegraphics[width=1in,height=1.25in,clip,keepaspectratio]{figures_final/Fiedler_downsized.png}}]{Christian Fiedler} 
Christian Fiedler received the B.Sc. degree in Mathematics from the University of Bayreuth in 2015, the M.Phil. degree in Computational Biology from the University of Cambridge in 2017, and the M.Sc. degree in Mathematics from the Technical University of Munich. During his studies, he was supported by the Max Weber program. In 2019, he joined the Max Planck Institute for Intelligent Systems and the University of Stuttgart as a PhD student, before transferring to RWTH Aachen University in 2021. He has broad research interests at the intersection of control theory and machine learning, with a particular focus on kernel methods.
\end{IEEEbiography}

\begin{IEEEbiography}[{\includegraphics[width=1in,height=1.25in,clip,keepaspectratio]{melanie_zeilinger.jpg}}]{Melanie N.\ Zeilinger} is an Associate Professor at ETH Zurich, Switzerland. She received the Diploma degree in engineering cybernetics from the University of Stuttgart, Germany, in 2006, and the Ph.D. degree with honors in electrical engineering from ETH Zurich, Switzerland, in 2011. From 2011 to 2012 she was a Postdoctoral Fellow with the Ecole Polytechnique Federale de Lausanne (EPFL), Switzerland. She was a Marie Curie Fellow and Postdoctoral Researcher with the Max Planck Institute for Intelligent Systems, Tübingen, Germany until 2015 and with the Department of Electrical Engineering and Computer Sciences at the University of California at Berkeley, CA, USA, from 2012 to 2014. From 2018 to 2019 she was a professor at the University of Freiburg, Germany. She was awarded the ETH medal for her PhD thesis, an SNF Professorship, the ETH Golden Owl for exceptional teaching in 2022 and the European Control Award in 2023. Her research interests include learning-based control with applications to robotics and human-in-the-loop control.
\end{IEEEbiography}

\begin{IEEEbiography}[{\includegraphics[width=1in,height=1.25in,clip,keepaspectratio]{sebastian_trimpe.JPG}}]{Sebastian Trimpe}  Sebastian Trimpe received the B. Sc. degree in general engineering science and the M. Sc. degree (Dipl.-Ing.) in electrical engineering from Hamburg University of Technology, Hamburg, Germany, in 2005 and 2007, respectively, and the Ph. D. degree (Dr. sc.) in mechanical engineering from ETH Zurich, Zurich, Switzerland, in 2013. Since 2020, he is a full professor at RWTH Aachen University, Germany, where he heads the Institute for Data Science in Mechanical Engineering. Before, he was an independent Research Group Leader at the Max Planck Institute for Intelligent Systems in Stuttgart and Tübingen, Germany. His main research interests are in systems and control theory, machine learning, networked systems and robotics. Dr. Trimpe has received several awards, including the triennial IFAC World Congress Interactive Paper Prize (2011), the Klaus Tschira Award for achievements in public understanding of science (2014), the Best Paper Award of the International Conference on Cyber-Physical Systems (2019), and the Future Prize by the Ewald Marquardt Stiftung for innovations in control engineering (2020).
\end{IEEEbiography}

\begin{IEEEbiography}[{\includegraphics[width=1in,height=1.25in,clip,keepaspectratio]{koehler.png}}]{Johannes Köhler} received his Master degree in Engineering Cybernetics from the University of Stuttgart, Germany, in 2017. 
In 2021, he obtained a Ph.D. in mechanical engineering, also from the University of Stuttgart,
Germany, for which he received the 2021 European Systems \& Control PhD award.
He is currently a postdoctoral researcher at  ETH Zürich. 
His research interests are in the area of model predictive control and control and estimation for nonlinear uncertain systems.
\end{IEEEbiography}